\documentclass[12pt]{amsart}
\usepackage[latin1]{inputenc}
\usepackage[T1]{fontenc}
\pdfoutput=1

\usepackage[square,numbers]{natbib}

\RequirePackage{amsmath}
\RequirePackage{amssymb}
\usepackage{dsfont}
\usepackage{multirow}
\usepackage{mathrsfs}
\usepackage{ upgreek }
\usepackage{float}
\usepackage{lmodern}
\usepackage{color}
\usepackage{graphics}
\usepackage{ragged2e}
\usepackage{array}
\usepackage{flexisym}
\usepackage{ bbold }
\usepackage{stmaryrd}
\usepackage{amsthm}
\usepackage{ bbm }
\usepackage{ stmaryrd }
\usepackage{ frcursive }
\usepackage{ comment }
\usepackage{pgf, tikz}
\usetikzlibrary{shapes}
\usepackage{varioref}
\usepackage{enumitem}
\usepackage{xfrac}

\numberwithin{equation}{section}

\definecolor{rouge}{rgb}{0.7,0.00,0.00}
\definecolor{vert}{rgb}{0.00,0.5,0.00}
\definecolor{bleu}{rgb}{0.00,0.00,0.8}

\newtheorem{theorem}{Theorem}[section]
\newtheorem{lemma}[theorem]{Lemma}

\labelformat{hypothesis}{\textbf{M\kern-0.1mm#1}}

\theoremstyle{definition}

\newtheorem{conditionC}{C\kern-0.mm}
\labelformat{conditionC}{\textbf{C\kern-0.mm#1}}

\usepackage[colorlinks=true, linkcolor=blue, citecolor=blue]{hyperref}
\newcommand*{\given}[2]{\left.#1\,\middle|\,#2\right.}

\begin{document}

\title{Estimation of extreme survival probabilities with Cox model}

\author{Ion Grama and K\'evin Jaun\^atre}
\email{Contact: K\'evin Jaun\^atre. Email: kevin.jaunatre@hotmail.fr}
\date{LMBA, University of South Brittany, Vannes, France}

\maketitle

\begin{abstract}
We propose an extension of the regular Cox's proportional hazards model which allows the estimation of the probabilities of rare events. It is known that when the data are heavily censored, the estimation of the tail of the survival distribution is not reliable.
To improve the estimate of the baseline survival function in the range of the largest observed data and to extend it outside, we 
adjust the tail of the baseline distribution  beyond some threshold 
by an extreme value model under appropriate assumptions. 
The survival distributions conditioned to the covariates are easily computed from the baseline. 
A procedure allowing an automatic choice of the threshold and an
aggregated estimate of the survival probabilities are also proposed. 
The performance is studied by simulations and
an application on two data sets is given. 

\end{abstract}

\section{Introduction}

The proportional hazards model introduced by \cite{Cox1972} has been largely 
studied over the years and multiple extensions have been made to the original model. 
These developments 
%include %However, the study of the survival data 
often aim to make inference on the regression parameters in various setting including
censoring, time-dependent coefficients, stratified and multistates models, missing and incomplete data etc. % \textcolor{magenta}{???}\todo{à ajouter un exemple}. 
% analysis of residuals,
The references \cite{fleming2011}, \cite{andersen2012}, \cite{therneau2013} and \cite{klein2005} give an overall view of this subject. 
%While many studies are focused on the properties of the regression parameter
%the estimation of the underlying hazard functions is an important ingredient which
%mostly follow the interrogation of knowing the effect of a treatment. 
%Nowadays, many studies are focused on the properties of the regression parameter. But t
The estimation of the underlying hazard functions is an important ingredient which mostly follows
the interrogation of knowing the effect of a treatment.
%a connection to the study of counting processes and martingale theory (\cite{fleming2011} and \cite{andersen2012}),
%%%%%%%%%%%%%%
%\todo{Introduction a am\'eliorer}
%%%%%%%%%%%%%%
We refer to \cite{crowley1977}, \cite{lee1992} among others for an illustration. 
If we are in presence of a significant amount of censored data it is well known that 
we cannot predict with sufficient precision how the tail of the estimated survival distribution behaves
%we don't know what is going on 
near or beyond the last observed value.
This is illustrated in the Figure \ref{IntBlad} (top) using a real data example,
where we can observe that the estimated baseline survival probability becomes constant in the long run.

The present paper aims to estimate the values of the
survival distributions conditionally to a covariate in the Cox's proportional hazards model
in the case when the estimated probabilities are out of the range of the observed data
by using extreme value modeling.     
%Extreme value modeling is used in many applications in different areas. 
%%%%%%%%%%%%%%
%\todo{A am\'eliorer}
%%%%%%%%%%%%%%
%the extreme value theory , 
Our analysis is based on the Peak-Over-Threshold method (see \cite{embrechts2013}%and \cite{BeirlantAl2004}
), 
which allows to estimate the tail of a distribution beyond a threshold. 
Applications of this method can be seen in various domain such as insurance, biology, ecology, whether forecast etc. 
For insurance and financial applications, we can refer to \cite{mcneil1997} and \cite{danielsson1997} among many others. 
A rainfall data study can be found in \cite{Gardes2010} and high-frequency oyster data are studied in \cite{DGPT2015}.
For  some results related to the estimation of extreme values under random censoring without covariates 
we refer to \cite{Beirlant2007}, \cite{Beirlant2010} and \cite{Einmahl2008}.
The case with covariates has been considered in \cite{Stupfler2016}, \cite{Ndao2014} and  \cite{Ndao2016}.
However, to the best of our knowledge, this approach was not used so far in the context of the Cox model with covariates. 
%However all these papers do not deal with the Cox model.

%The use of the extreme value with the Cox model with covariates seems to encounter some difficulties, as we will see latter. 

Our idea is simple: we adjust a Pareto distribution for observations beyond a threshold, while 
the remaining part is estimated nonparametrically by the Nelson-Aalen estimator.  
The main difficulty is the appropriate choice of the threshold, which can be problematic 
as a large value will lead to an important variability and small value will increase the bias. 
This quandary
%\todo{Maybe quandary is better than dilemma here?} 
is well known in theory of extreme values. %we refer to  \cite{embrechts2013} and \cite{BeirlantAl2004}. 
To choose the threshold we make use of the adaptive procedure based on consecutive 
tests developed in \cite{GramaSpokoiny2008}.  %and \cite{DGPT2015}. 
In addition to this, we propose an aggregation procedure which allows to improve significantly the stability of the estimation. 
The performance of the proposed estimators is demonstrated by a simulation study and some applications are given. 
As an illustration in Figure \ref{IntBlad} (bottom) we show the estimated baseline survival function by
the proposed method.

The paper \cite{grama2014} deals also with the Cox model but in a very different way.
In \cite{grama2014}, the Cox model with qualitative covariates, say $z_i$ with a finite number of modalities is considered.
Each modality has its own choice of the threshold $\widehat\tau$, which in principle are different.
In the present paper we estimate only one threshold $\widehat\tau$ under the baseline distribution, which is automatically translated to the other modalities of the covariates. Thus we can deal with any type of covariates, and we have only one choice of the threshold.
We also note that the paper \cite{DGPT2015} deals with a non stationary time series setting, where the goal is to estimate high quantiles driven by a non-parametricaly changing distribution function. 

The paper is organized as follows.
In Section \ref{modEst}, we introduce the notations, formulate the  model and we state the main results.
An explicit computation of the convergence rate using the Hall model is 
given as an example  in Section \ref{Hall}.
An automatic selection procedure of the threshold is stated in Section \ref{autsel}.
In Section \ref{Aggreg} we formulate our procedure for the aggregation of the estimated survival probabilities.
A simulation study is done in Section \ref{sim} and an application on two data sets is given in Section \ref{app}.

\section{Main results}\label{modEst}

\subsection{Notations and model}

Denote by $X$ a random variable representing the failure time, 
by $C$ a random variable representing the censoring time 
and  
by $Z$ a random covariate vector.  %for individual $i$, where $i=1,...,n$. 
We assume that $X$ and $C$ admit positive densities on $[x_0,\infty)$, with $x_0\geq 0$ and that $X$ and $C$ are independent conditionally to $Z$. 
The observation time and failure indicator are respectively
$$
T = \min\{X,C \} \quad \text{and} \quad \Delta = \mathbb{1}_{X\leq C},
$$
where $\mathbb{1}$ is the indicator function.
The Cox model (introduced by \cite{Cox1972}) specifies that the hazard function of the failure time $X$ depends on the value $z$ of covariate vector $Z$  as follows: 
\begin{equation*}
h(\given{x}{z}) = \exp(\beta \cdot z) h_0(x), \ x \geq x_0,
\end{equation*}
where $x_0\geq 0,$ $\beta$ is a vector of parameters, 
$h_0$ is an unknown baseline hazard function %giving the hazard function for $z = 0$ 
and $ \beta \cdot z$ denotes the scalar product between $\beta$ and $z$. 
We denote by $f(\given{x}{z})$ and $S(\given{x}{z}) = 1 - F(\given{x}{z})$, respectively the density and survival functions of the failure time $X$ given $Z=z$. 
The hazard function $h(\given{x}{z})$ is related to the functions $f(\given{x}{z})$ and $S(\given{x}{z})$ %the density and survival functions 
by the expressions 
$$h(\given{x}{z})=f(\given{x}{z})/S(\given{x}{z})$$ 
and 
$$S(\given{x}{z}) = \exp(-\int_{x_0}^{x}h(\given{u}{z}) du).$$ 
Similarly, the hazard function of the censoring time is denoted by $h_C(\cdot)$, 
the density and survival function are respectively denoted $f_C(\cdot)$ and $S_C(\cdot)= 1 - F_C(\cdot)$.
Let 
\begin{align} \label{defS_0}
S_0(x) = S(\given{x}{0}) = \exp(-\int_{x_0}^{x}h_0(u) du)
\end{align}
be the baseline survival function. 
The survival function $S(\given{\cdot}{z})$ is related to the baseline survival function $S_0(\cdot)$ by the expression 
$$S(\given{x}{z}) = S_0(x)^{\exp(\beta \cdot z)}.$$

Assume that we observe a sequence of independent triples $(t_i,\delta_i,z_i)$, $i=1,...,n$, where all $z_i$'s are nonrandom and each pair $(t_i,\delta_i)$ has the law of $(T,\Delta)$ given $Z=z_i$. 
In this paper we address the question of estimating the survival function $S(\given{x}{z}) $ for large values of $x$.

Let us explain the difficulties related to this problem using the classical Nelson-Aalen estimator.
In the case when $x$ is larger than the last observed time $t_{\max}=\max\left\{ t_1,...,t_n \right\}$,  the Nelson-Aalen estimator $\widehat S_{NA}(\given{\cdot}{z})$
of $S(\given{\cdot}{z})$ takes two positive constant values depending on the fact that the last observed time is censored or not. 
In Figure \ref{IntBlad} (top) we plot the estimated baseline survival function $\widehat S_{NA}(\given{\cdot}{0})$
for the commonly accessible \verb+bladder+ data set from \verb+R+ package \verb+survival+ (see Section \ref{appBlad} for details).  
\begin{figure}[h]
	\begin{center} 
		\begin{tabular}{c}
			\includegraphics[width=120mm,height=60mm]{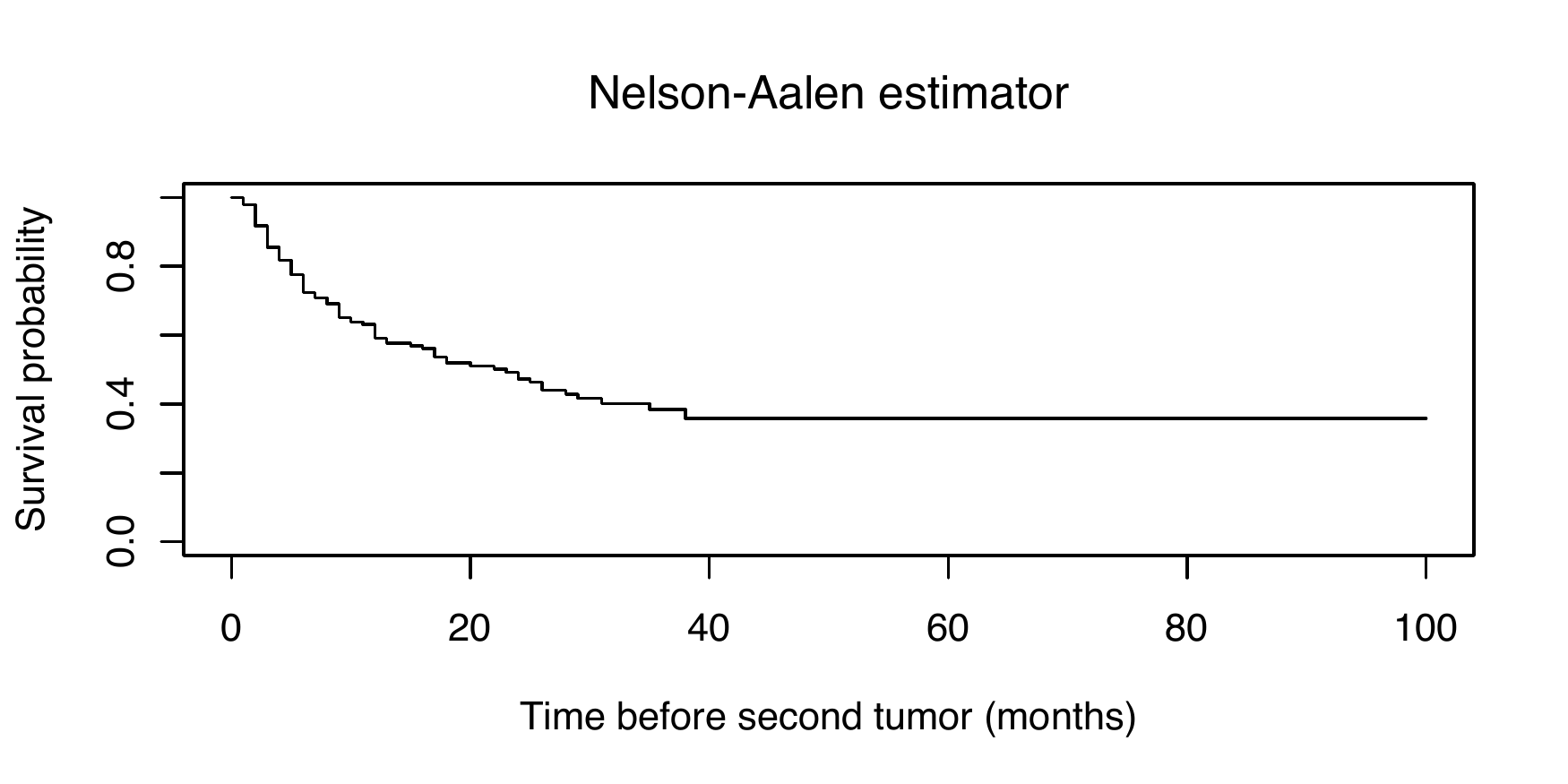}\\
			\includegraphics[width=120mm,height=60mm]{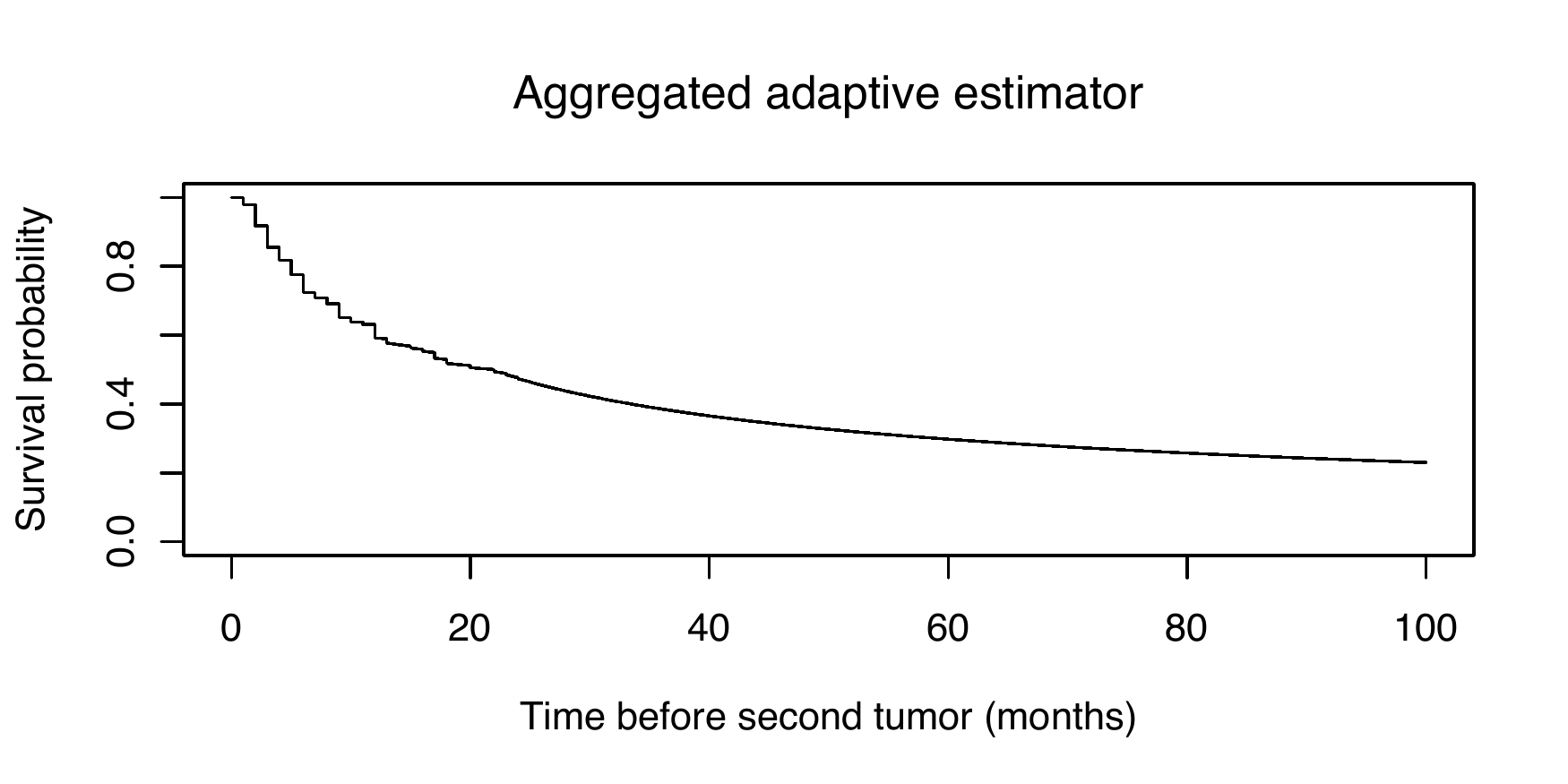}
		\end{tabular}
		\caption{
		        Estimated baseline survival probability (no-treatment) of the time of having the first recurrence of the bladder tumor:
			using the Nelson-Aalen estimator defined by \eqref{NAestim} (top) and the adaptive 
			aggregation defined by \eqref{Aae} (bottom).
			}
		\label{IntBlad}
	\end{center}
\end{figure}
Note that  the Nelson-Aalen curve $\widehat S_{NA}(\given{t}{0})$ becomes constant for $t \geq t_{\max} $. 
Moreover, as the survival times are heavily censored, the estimated survival probability 
$\widehat S_{NA}(\given{t}{0})$ is far above $0$ for all $t \geq t_{\max} $.
To overcome this effect, we assume some additional constraints on the survival function, which allow us to 
extrapolate it outside the available data range.
Specifically, we assume the following condition:
\begin{conditionC} \label{CFrech}
	We assume that $F_0$ belongs to the maximal domain of attraction of the Fr\'echet law with extreme value index $\theta>0$ 
	which means that
	there exists two sequences $a_n>0$ and  $b_n$ such that, for any $x\geq 0$,
	$$
	F_0^n(a_nx+b_n) \to \Phi_{\theta}(x) \ \mbox{as} \ n\to\infty,
	$$
	where 
	$\Phi_{\theta}(x)= e^{-x^{-\theta}},$ $x\geq 0$ is the Fr\'echet law 
	and
	$F_0^n(x) = \mathbb P ( \max_{1\leq i \leq n}{x_i} \leq x)$, 
	with $x_1,x_2,\dots$ an i.i.d.\ sequence of random variables of common distribution $F_0$.
	%there exists two sequences $a_n>0$ and  $b_n$ such that
	%$
	%\frac{M_n-b_n}{a_n} 
	%$,
	%with $M_n=\max_{1\leq i \leq n}{x_i}$,  
	%converges in distribution to the Fr\'echet law $\Phi_{\theta}(x)= e^{-x^{-\theta}}$, where $x_1,X_2,\dots$ is a sequence 
	%of i.i.d.\ random variables with law $F_0.$
	%The baseline distribution function $F_0=1-S_0$ belongs to the domain of attraction of the Fr\'echet distribution. 
	% with parameter $1/\theta$, for some $\theta>0.$
\end{conditionC}
By the Fisher-Tippett-Gnedenko theorem (see Theorem 2.1 page 75 in \cite{BeirlantAl2004}), condition \eqref{CFrech} is equivalent to the property that 
for each $x\geq1$, %and a threshold $\tau>x_0$,
\begin{equation}\label{FTP theorem}
\frac{S_0(\tau x)}{S_0(\tau)} \to \left(x\right)^{-1/\theta}\ \text{as} \ \tau\to\infty. 
\end{equation}
%Motivated by \eqref{FTP theorem},  we can adjust the baseline excess distribution function 
%$F_{0,\tau}(x) = 1- \frac{1-F_0(x)}{1-F_0(\tau)} $, 
%$ x \in [\tau,\infty), $
%%beyond some threshold $\tau\geq x_0$
%by a Pareto distribution with parameter $\theta>0$ when the threshold $\tau\geq x_0$ is large enough.
%Motivated by \eqref{CFrech} and by the Fisher-Tippett-Gnedenko theorem (see Theorem 2.1 page 75 in \cite{BeirlantAl2004}), we can adjust the baseline excess distribution function 
%$F_{0,\tau}(x) = 1- \frac{1-F_0(x)}{1-F_0(\tau)} $, 
%$ x \in [\tau,\infty), $
%beyond some threshold $\tau\geq x_0$
%by a Pareto law with parameter $\theta>0$ when the threshold $\tau\geq x_0$ is large enough.
%To overcome this effect, we take advantage of the extreme value property \eqref{FTP theorem} which ensures that the tail of $S_0(x)$ 
%can merely be fitted by a Pareto distribution with parameter  $\theta>0$. 
As a consequence of \eqref{FTP theorem} the following semi-parametric model is considered for the baseline survival function: %when $z=0$:
\begin{equation}\label{sem-mod}
S_{0,\tau,\theta}(x) =
\left\{
\begin{array}{ll}
S_{0}\left( x\right) & \text{if\ \ }x\in \lbrack 0,\tau ], \\
S_{0}\left( \tau \right) \left( \frac{x}{\tau} \right)^{-1/\theta} & \text{if\ \ }x>\tau ,%
\end{array}\right.
\end{equation}
where the function $S_0$ is fully non-parametric,
$\tau$ is a threshold parameter and the parametric part is completely described by the Pareto model with parameter $\theta.$
We denote the baseline hazard function of the previous model by 
%$$h_0(\given{x}{\tau,\theta}) = $$
\begin{equation} \label{hazard h0 001}
h_{0,\tau,\theta} (x) = \left\{
\begin{array}{ll}
h_{0}\left( x\right) & \text{if\ \ }x\in \lbrack 0,\tau ], \\
\frac{1}{\theta x} %h_{0}\left( \tau \right) \left( \frac{x}{\tau} \right)^{-\frac{1}{\theta}} 
& \text{if\ \ }x>\tau .%
\end{array}\right.
\end{equation}
The corresponding cumulative hazard and survival function under the covariate constraint $Z=z$ are then respectively given by $H_{z,\tau,\theta} (x) = e^{\beta \cdot z}\int_{x_0}^{x}h_{0}\left( x\right) $ and
\begin{align} \label{themodel001}
S_{z,\tau,\theta} (x) = S(\given{x}{z,\tau,\theta})= S_{0,\tau,\theta}(x)^{e^{\beta \cdot z}}.
\end{align} 
%and the associated distribution function by 
%\begin{align} \label{themodel002}
%S_{0,\tau,\theta}(\given{x}{z})=1-S_{\theta,\tau}(\given{x}{z}).
%\end{align}
For illustration purposes the estimated baseline survival function $S_0$ by
the proposed model is given in Figure \ref{IntBlad} (bottom), 
where for the estimation we have used the aggregation procedure with the adaptive choice of the threshold $\tau$
described in Section \ref{AddAggreg}.  

\subsection{Estimators}
In this section, we aim to provide the estimators necessary to deal with the model \eqref{sem-mod}. 
To this end, we suppose the regression parameter $\beta$ of the Cox model known,
for example, estimated by the standard procedure described in \cite{Cox1972}.
As to the threshold $\tau$, it is considered to be fixed  
(a selection procedure is presented latter on in Section \ref{autsel}).
To estimate \eqref{sem-mod}, we will combine  the extreme value Hill type estimator of the parameter $\theta$ of the tail
of the distribution function $F_0$ and 
the Nelson-Aalen non-parametric estimator of the baseline survival function $S_0.$

The joint density of the vector $(T,\Delta)$, given $Z=z$, is computed as %  $P_{S_0}(\given{dt,d\delta}{z})$ is 
\begin{equation}\label{dens}
p_{S_0}(\given{t,\delta}{z}) = (e^{\beta \cdot z} h_0(t))^\delta S_0(t)^{e^{\beta \cdot z}} f_C(t)^{1-\delta}S_C(t)^\delta,
\end{equation}
where $t \in [ x_0, \infty )$ and $\delta = \{0,1 \}$. 
Denote by $p_{S_{0,\tau,\theta}}(\given{t,\delta}{z}) $
the joint density of the vector $(T,\Delta)$, given $Z=z$, when 
the survival function $S(\given{t}{z})$ obeys the model 
%\eqref{themodel002} and 
\eqref{themodel001}:
%\eqref{sem-mod}:
$$
p_{S_{0,\tau,\theta}}(\given{t,\delta}{z}) = (e^{\beta \cdot z}h_{0,\tau,\theta} (t))^\delta S_{0,\tau,\theta}(t)^{e^{\beta \cdot z}}f_C(t)^{1-\delta}S_C(t)^\delta.
$$
%The corresponding conditional cumulative distribution function of the vector $( T, \Delta )$ given $Z=z$ is denoted by $P_{S_0}(\given{t,\delta}{z}).$  
Then the quasi-log-likelihood of the model is 
\begin{eqnarray*}
\mathcal{L}(\given{\theta}{\mathbf{z}})=\mathcal{L}(\given{\theta}{z_1,...,z_n}) 
= \sum_{i=1}^n \ln p_{S_{0,\tau,\theta}}(\given{t_i,\delta_i}{z_i}).
\end{eqnarray*}
Removing the terms related to the censoring, the partial quasi-log-likelihood is 
\begin{equation*}
\mathcal{L}^{part}(\given{\theta}{\mathbf{z}}) = \sum_{i=1}^n \left( \delta_i \ln(h_{0,\tau,\theta} (t_i)) + (\beta \cdot z_i )\delta_i - e^{\beta \cdot z_i} \int_{x_0}^{t_i} h_{0,\tau,\theta} (u)du \right), 
\end{equation*}
where the baseline hazard function $h_{0,\tau,\theta} (.)$ 
is defined by \eqref{hazard h0 001} and $\mathbf{z}= z_1, \cdots, z_n$.
%denotes  of the model (\ref{sem-mod}). 
Now, maximizing $\mathcal{L}^{part}(\given{\theta}{\mathbf{z}})$ with respect to $\theta$, yields the estimator 
\begin{equation}
\label{estTheta}
\widehat\theta_\tau = \frac{\sum_{t_i>\tau}e^{\beta \cdot z_i}\ln\left( \frac{t_i}{\tau}\right) }{\sum_{t_i>\tau}\delta_i}.
\end{equation}
One can see that, the estimator is a transformation of the estimator introduced in \cite{hill1975}.

The Nelson-Aalen estimator, suggested by \cite{nelson1972} and rediscovered by \cite{aalen1978}, focus on estimating the cumulative %(or integrated) 
hazard function $H_0(x) = \int_0^x h_0(t)dt$ by
\begin{equation*}
\widehat{H}_0(t)= \sum_{t_i\leq t} \widehat h_0(t_i),
\end{equation*}
%$$\widehat{H}_0(t)=\left\{ 
%\begin{array}{cc}
%0 & \text{if\ \ }t<t_1, \\
%\sum_{t_i\leq t} \widehat h_0(t_i) & \text{if\ \ }t\geq t_1 ,%
%\end{array}\right.
%$$
where, %$\widehat{h}_0(t_1),...,\widehat{h}_0(t_n)$ are found  
by maximizing the partial quasi-log-likelihood, we have %with respect to $h_0(t_i)$'s
\begin{equation*}
\widehat{h}_{0} (t_i) = \frac{\delta_i}{\sum_{j=1}^n e^{\beta \cdot z_j}\mathds{1}_{t_j \geq t_i} }.
\end{equation*}
This estimator was suggested by \cite{breslow1972}.
The estimator of the survival function (\ref{sem-mod}) is then given by
\begin{equation}
\label{EMSP}
\widehat{S}_{0,\tau,\widehat\theta_\tau}(t) = \left\{
\begin{array}{cc}
\widehat{S}_{0}\left( t\right) & \text{if\ \ }t\in \lbrack 0,\tau ], \\
\widehat{S}_{0}\left( \tau \right) \left( \frac{t}{\tau} \right)^{-1/\widehat{\theta}_\tau} & \text{if\ \ }t>\tau ,%
\end{array}\right.
\end{equation}
where the Nelson-Aalen non-parametric estimator of $S_{0}$ is defined by
\begin{align} \label{NAestim}
\widehat{S_{0}}\left( t\right) = e^{-\widehat H_0(t)}. %  = e^{-\sum_{t_j\leq t}\widehat{h}_0(t_j) }.
\end{align}
The estimator of the survival function \eqref{themodel001} is then given by $\widehat{S}_{z,\tau,\theta} (t)  = \widehat{S}_{0,\tau,\widehat\theta_\tau}(t)^{e^{\beta \cdot z}}$.

\subsection{Consistency of $\widehat\theta_\tau$}\label{ConsTheta}
In this section we state a general consistency result for the estimator $\widehat\theta_\tau$ of $\theta$, 
which we apply in Section \ref{Hall} to obtain the rate of convergence under the Hall model. 
To state it, we need some notations.

The Kullback-Leibler divergence 
between two equivalent distributions, say $P$ and $Q$, is denoted
$\mathcal{K}(P,Q)=\int \ln(dP/dQ)dP.$  
This divergence, between two Pareto distributions with parameters $\theta'$ and $\theta$, can be written as 
\begin{equation}\label{KL entropy}
\mathcal{K}(\theta',\theta) = \frac{\theta'}{\theta} - 1 - \ln(\frac{\theta'}{\theta})
\sim \left(  \frac{\theta'}{\theta}-1\right)^2
\quad\text{as} \quad \frac{\theta'}{\theta}\to1.
\end{equation}
The $\chi^2$-entropy between the probability measures $P$ and $Q$ is defined by 
\begin{equation}\label{Xi2 entropy}
\chi^2 (P,Q) = \int dP/dQdP -1.
\end{equation}
%In the following $a_n = O_{\mathbb P} (b_n)$ means $\mathbb{P}(a_n>cb_n,b_n < \infty) \rightarrow 0$ as $n \rightarrow \infty$ and $c$ is a positive constant.
%For any $\theta>0$, the optimal rate of convergence is obtained when $\tau$ is chosen such that 
%\begin{equation}
%\label{c1}
% \chi^2(P_{S_0}(\given{.}{z}),P_{F_\theta}(\given{.}{z})) = O\left( \frac{\ln n}{n} \right),
%\end{equation}
%In order to define the rate of convergence of the estimator $\widehat\theta_\tau$, we need one assumption. The assumption involves 
%We also need a control of the proportion of censored observations among the all the observations exceeding a threshold $\tau.$   
The following theorem gives an estimate of the Kullback-Leibler entropy between $\widehat \theta_{\tau}$ and  $\theta$ which is expressed in terms 
of the $\chi^2$-entropy between the two laws $P_{S_0}(\given{.}{z_i})$ and $P_{S_{0,\tau,\theta}}(\given{.}{z_i})$. 
The notation $a_n = O(b_n)$ means that there is a positive constant $c$ such that $\mathbb{P}(a_n>cb_n,b_n<\infty)\rightarrow 0$ as $n \rightarrow \infty$. In the sequel we denote by $\mathbb P$ the probability measure corresponding to the "true" model which has 
the baseline survival function $S_0$.

\begin{theorem}
	\label{thp1} 
	Assume that $S_0$ is an arbitrary survival function as defined by \eqref{defS_0} and that $S_{0,\tau,\theta}$ satisfies
	the model \eqref{sem-mod}. Then for any $\theta>0$, and $\tau\geq x_0$, we have 
	$$
	\mathcal{K}(\widehat \theta_\tau,\theta) = O_{\mathbb{P}}\left( \frac{1}{\widehat n_\tau} \sum_{i=1}^n\chi^2(P_{S_0}(\given{.}{z_i}),P_{S_{0,\tau,\theta}}(\given{.}{z_i})) + \frac{4\ln(n)}{\widehat n_\tau} \right),
	$$
	where $\widehat n_\tau = \sum_{t_i>\tau}\delta_i$.
\end{theorem}
Theorem \ref{thp1} gives an upper bound of the Kullback-Leibler entropy which can be read in two parts, the bias term 
$\frac{1}{\widehat n_\tau}\sum_{i=1}^n \chi^2(P_{S_0}(\given{.}{z_i}),P_{S_{0,\tau,\theta}}(\given{.}{z_i}))$ 
and the variance term $\frac{4\ln(n)}{\widehat n_\tau}$. 

Now, we formulate some sufficient conditions for the consistence of the estimator $\widehat \theta_\tau$.
In order to estimate the bias term, we need to introduce a quantity which show how the censoring rate evolves with the threshold $\tau$, 
when we consider only the observations %$\left\{ T_i: T_i\geq \tau\right\}$ 
exceeding $\tau.$  
For this, we introduce the following conditioned mean censoring rate function given $Z=z$: 
%and the observations $\left\{ T_i: T_i\geq \tau\right\},$ %above the threshold $T \geq \tau,$ 
\begin{equation*}
%\label{ccrf}
\tau \mapsto q_{F}(\given{\tau}{z}) = \int_\tau^{\infty} \frac{S(\given{t}{z})}{S(\given{\tau}{z})} \frac{f_C(\given{t}{z})}{S_C(\given{\tau}{z})}dt \in [0,1],
\quad \tau \geq x_0.
\end{equation*}
%For a given $z$, 
The value $q_{F}(\given{\tau}{z})$  gives the rate of censored observations above the threshold $\tau$.
In order to estimate the extreme survival probabilities it is natural to require that the
rate of censored observations above the threshold $\tau$ is strictly less than $1$.  
We shall impose on $q_{F}(\given{\tau}{z})$ the following condition: 
\begin{conditionC} \label{Crate}
	There is a constant $q_0<1$ such that, for $\tau \geq x_0$ large enough,
	\begin{equation*} %\label{c3}
	q_{F}(\given{\tau}{z_i}) \leq q_0, \quad i=1,...,n.
	\end{equation*}
\end{conditionC}
Condition \eqref{Crate} is easily verified, for instance, when both the  %baseline 
distribution function of the survival time and that of the censoring time follow the Cox model 
and are in the maximal domain of attraction of the Fr\'echet law with parameters $\theta(z)$ and $\theta_C(z)$ respectively.
% see conditions \eqref{cHall1} and \eqref{cHall2} below.
Indeed, we show in the Lemma \ref{LCrate} that, in this case, for any $z$,
\begin{align*} %\label{}
q_{F}(\given{\tau}{z}) \to \frac{\theta(z)}{\theta(z) + \theta_C(z)} \ \mbox{as} \   \tau \to\infty.
\end{align*}
where under some mild assumptions one can verify that 
\begin{align} \label{censoringcond001}
\frac{\theta(z)}{\theta(z) + \theta_C(z)} \leq q_0<1.
\end{align}

In addition we introduce the following conditions: % are necessary to state the result of Theorem \ref{th2}.
\begin{conditionC} \label{CBound}
	There exists $z_{\min}$ and $z_{\max}$ %$\beta_{\min}$ and $\beta_{\max}$ 
	such that
	\begin{equation*}
	%\label{c4}
	z \in \mathbb{Z} := [z_{\min}; z_{\max}]. %\quad \text{and} \quad \beta \in [\beta_{\min}; \beta_{\max}] 
	\end{equation*}
\end{conditionC}
\begin{conditionC}  Von-Mises condition:
	\label{CvonMises}
	\begin{equation*} %\label{cVM}
	th_0(t) \rightarrow \frac{1}{\theta} \quad \text{as}  \quad t \rightarrow \infty.
	\end{equation*}
\end{conditionC}
It is well known that \eqref{CvonMises} implies condition \eqref{CFrech}, 
see \cite{BeirlantAl2004}.
For any $\tau>0$, set
$$
\rho_{\tau} = \sup_{t > \tau} \left| th_0(t) - \frac{1}{\theta} \right|.
$$
It is easy to see that the Von-Mises condition \eqref{CvonMises} is equivalent to the fact 
that there exists a sequence of thresholds $(\tau_n)$ satisfying $x_0 \leq \tau_n\to \infty$ as $n\to\infty$ such that
$n\rho_{\tau_n}^2 \rightarrow 0$ as $n \rightarrow \infty$. 
%\begin{conditionC} \label{CLimRho}
%	%There exists a sequence $(\tau_n)$ satisfying $x_0 \leq \tau_n\to \infty$ as $n\to\infty$ such that
%	\begin{equation*} \label{condMises00A}
%	n\rho_{\tau_n}^2 \rightarrow 0 \quad \text{as} \quad n \rightarrow \infty.
%	\end{equation*}
%\end{conditionC} 
%We need in addition that the sequence $(\tau_n)$ from condition \eqref{CLimRho} satisfies:  
%\begin{conditionC} \label{CLimSurv}
%	\begin{equation*} \label{denominator001}
%	\sum_{i=1}^nS(\given{\tau_n}{z_i}) S_C(\given{\tau_n}{z_i}) \rightarrow \infty \quad \text{as} \quad n \rightarrow \infty. 
%	\end{equation*}
%\end{conditionC}

The following result shows the consistency of the estimated parameter $\widehat \theta_{\tau_n}$: % under the conditions stated above.
\begin{theorem}
	\label{th2}
	%Let $F_0=1-S_0$ be a distribution function in the domain of attraction of the Fr\'echet distribution.  Assume that there exist a constant $q_0 \in (0,1)$ 
	Assume conditions \eqref{Crate}, \eqref{CBound} and \eqref{CvonMises}. 
	For any sequence $(\tau_n)$ satisfying 	%$x_0 \leq \tau_n\to \infty$ as $n\to\infty$
	%	\begin{equation} \label{seqtau001}
	%	 \quad \text{as} \quad n \rightarrow \infty,
	%	\end{equation}	  
	\begin{equation} \label{condMises00A}
	x_0 \leq \tau_n\to \infty, \quad n\rho_{\tau_n}^2 \rightarrow 0 \quad \text{as} \quad n \rightarrow \infty
	\end{equation}
	and 	\begin{equation} \label{denominator001}
	\sum_{i=1}^nS(\given{\tau_n}{z_i}) S_C(\given{\tau_n}{z_i}) \rightarrow \infty \quad \text{as} \quad n \rightarrow \infty, 
	\end{equation}
	it holds
	$$
	\widehat \theta_{\tau_n} \xrightarrow[n\rightarrow \infty]{\mathbb{P}} \theta .
	$$
\end{theorem}
Condition \eqref{condMises00A} gives a control on the bias
$\frac{1}{\widehat n_\tau}\sum_{i=1}^n \chi^2(P_{S_0}(\given{.}{z_i}),P_{S_{0,\tau,\theta}}(\given{.}{z_i}))$
of the model which should be small as $n\to \infty$, 
while condition \eqref{denominator001} is responsible for the control of the variance $\frac{4\ln(n)}{\widehat n_\tau}$ 
of the model.  
We will show in the next section how to verify conditions \eqref{Crate},  \eqref{condMises00A} and \eqref{denominator001} with a large class of models.

\section{Computation of the explicit rate of convergence for the Hall model}\label{Hall}

In this section we consider a model which is related to the families of distributions 
in \cite{hall1982}, \cite{hall1984} and \cite{GramaSpokoiny2008} for the extreme value estimation. 
The result of the Theorem \ref{th1} in Section \ref{ConsTheta} shows that the rate of convergence of the estimator $\widehat\theta_{\tau_n}$ depends on the threshold $\tau_n$ and the survival functions of the survival and censoring times. 
To express the rate of convergence in terms of the sample size, some assumptions must be made on the survival functions $S$ and $S_C$. 
\begin{conditionC}\label{cHall1}
	The baseline hazard function $h_0$ is such that for some unknown parameter $\theta \in (\theta_{min},\theta_{max})$ and any $t>1,$
	\begin{equation*}
	|t h_0(t) - \frac{1}{\theta}|\leq c_1t^{-\frac{\alpha}{\theta}},
	\end{equation*}
	where $\alpha$, $c_1$, $\theta_{min}$ and $\theta_{max}$ are some positive constants. 
	%	\begin{equation*}
	%	|th(\given{t}{z}) - \frac{e^{\beta \cdot z}}{\theta}|\leq c_1t^{-\frac{\alpha e^{\beta \cdot z}}{\theta}},
	%	\end{equation*}
	%	where $\alpha$, $c_1$, $\theta_{min}$ and $\theta_{max}$ are some positive constants. 
\end{conditionC}
Condition \eqref{cHall1} means that $th_0(t)$ converges to $\frac{1}{\theta}$ polynomially fast as $t \rightarrow \infty$. 
Similarly we assume:
\begin{conditionC}\label{cHall2}
	The hazard function $h_C$ of the censoring time $C$ is such that for any $z \in \mathbb{Z} := [z_{\min}; z_{\max}]$ and any $t>1,$ 
	\begin{equation*}
	|th_C(\given{t}{z}) - \frac{1}{\theta_C(z)}|\leq c_2t^{-\mu},
	\end{equation*}
	where $\theta_C(z)= \frac{\theta}{\gamma}  e^{-\beta \cdot z}$ and $\gamma >0$, $c_2>0$ and $\mu>1$ are some constants.
\end{conditionC}
%\todo{Add a commentary on \eqref{cHall2}}
%\begin{theorem}
%\label{Th4}
%Assume condition (\ref{c3}) and that $h_0(x)$ satisfies (\ref{cHall1}), there is some positive constants $c''$ and $\tau_{min}$ such that $h_C(x) \geq c''$ for any $x>\tau_{min}$ and that 
%$$
%S_C(\tau)  n^{1-\frac{3}{4\theta \alpha+ 1}} \ln^{\frac{3}{4\theta \alpha+ 1}} (n) \rightarrow \infty, \quad as \quad n \rightarrow \infty.
%$$
%Then, there exists a constant $C_2$ such that, 
%$$
%\lim_{n\to \infty}  \mathbb{P}  \left(\mathcal{K}(\widehat \theta_{\tau_n},\theta)  \leq C_2 \frac{\left(n^{-1} \ln(n) \right)^{1-\frac{3}{4\theta \alpha+ 1}}}{S_C(\tau_n)}     \right) = 1,
%$$
%where
%$$
%\tau_n  =  \left(\frac{\ln(n)}{n}\right)^{-\frac{1}{\frac{1}{2\theta}+2\alpha}}.
%$$
%\end{theorem}
%Under the assumption that $h_C(x)$ satisfies (\ref{cHall2}), we have the result :
\begin{theorem}
	\label{Th5}
	Assume condition \eqref{CBound}, \eqref{cHall1} and \eqref{cHall2}. 
	Suppose in addition that $\frac{\varrho_2}{\varrho_1}  \frac{ 1 + \gamma} { 1+\gamma+2\alpha}  \leq 1,$
	where %$\varrho =\varrho_1 /\varrho = %$\varrho = \max_{\beta \cdot z}(e^{\beta \cdot z})/ \min_{\beta \cdot z}(e^{\beta \cdot z}) $ 
	$\varrho_1 = \min_{\beta \cdot z} (e^{\beta \cdot z})$ and $\varrho_2 = \max_{\beta \cdot z} (e^{\beta \cdot z})$.
	Then, there exists a constant $c$ such that, 
	$$
	\lim_{n\to \infty}  \mathbb{P}  \left(\mathcal{K}(\widehat \theta_{\tau_n},\theta)  
	\leq c \left( \frac{\ln n}{n} \right)^{ 1-\frac{\varrho_2}{\varrho_1}  \frac{ 1 + \gamma} { 1+\gamma+2\alpha}} \right)  = 1,
	%	\leq c \left( \frac{\ln n}{n} \right)^{\frac{\max(e^{\beta \cdot z}) 2\alpha}{\min(e^{\beta \cdot z}) (1+\gamma+2\alpha)}} \right)  = 1,
	$$
	where
	$$
	\tau_n = n^{\frac{  \theta/\varrho_1}{1+\gamma+2\alpha}} \ln^{-\frac{\theta/\varrho_1}{1+\gamma +2\alpha}} n.
	$$
\end{theorem}
%\todo{$\frac{\varrho_2}{\varrho_1}$ can be = to $10$. What to do ???}
When the covariate $z$ is absent, say $z=0$, then $\varrho_2=\varrho_1 = 1$ and 
we recover the result of Theorem 4.2 of the paper \cite{grama2014},
where it is shown that when $\gamma$ goes to $0$ (no censoring) the rate becomes close to
the optimal rate of convergence $n^{\frac{ 2\alpha} { 1+2\alpha}}$ in the context of the extreme value estimation, 
see \cite{drees1998} and \cite{GramaSpokoiny2008}. 
% When the values of the covariates belong to a shrinking interval containing one point $z_0$ such that 
% $\varrho_2/\varrho_1 \to 1$, % $z_{\max}-z_{\min} \to 0$ as
% we recover the result of the paper %he results found with Hall model coincide with the ones in the article 
% \cite{grama2014}.
In the case of a binary covariate (i.e. $z \in \{0,1\}$), if we assume that $\beta > 0$, the convergence speed becomes $\mathcal{K}(\widehat \theta_{\tau_n},\theta) = O_{\mathbb{P}}\left((\frac{\ln n}{n} )^{ 1- e^{\beta} \frac{ 1 + \gamma} { 1+\gamma+2\alpha}}\right)$ with $\tau_n = n^{\frac{  \theta}{1+\gamma+2\alpha}} \ln^{-\frac{\theta}{1+\gamma +2\alpha}} n.$

Condition \eqref{cHall2} may seem a bit cumbersome at first sight, however, with a closer look we see that it is quite natural
if we want to obtain a close rate of convergence.  
To see this we note that condition \eqref{cHall1} can be equivalently stated as
\begin{equation*}
|th(\given{t}{z}) - \frac{1}{\theta(z)}|\leq c_1t^{-\frac{\alpha}{\theta(z)}},
\end{equation*}
%where $\alpha$, $c_1$, $\theta_{min}$ and $\theta_{max}$ are some positive constants. 
where the tail index  $\theta(z)=\theta e^{-\beta \cdot z}$ depends on the covariate $z$. 
Now conditions \eqref{censoringcond001} and \eqref{Crate} will be verified with $q_0=\frac{\gamma}{\gamma+1}$.
%since
%\begin{align} \label{censoringcond002}
%\frac{\theta(z)}{\theta(z) + \theta_C(z)} = \frac{1}{1 + \theta_C(z)}   \leq q_0<1.
%\end{align}

Condition \eqref{cHall2} can be replaced by the following condition:
\begin{conditionC}\label{cHall2bis}
	The hazard function $h_C$ of the censoring time $C$ is such that for any $z \in \mathbb{Z} := [z_{\min}; z_{\max}]$ and any $t>1,$ 
	\begin{equation*}
	|th_C(t) - \frac{\gamma}{\theta}|\leq c_2t^{-\mu},
	\end{equation*}
	where %$\theta_C= \frac{\theta}{\gamma}$ and 
	$\gamma >0$, $c_2>0$ and $\mu>1$ are some constants.
\end{conditionC}
With \eqref{cHall2bis}  instead of \eqref{cHall2} we can obtain the following result:
\begin{theorem}
	\label{Th6}
	Assume condition \eqref{CBound}, \eqref{cHall1} and \eqref{cHall2bis}. 
	Suppose in addition that $\frac{ \varrho_2 + \gamma} { \varrho_1+\gamma+2\alpha\varrho_1}  \leq 1,$
	where %$\varrho =\varrho_1 /\varrho = %$\varrho = \max_{\beta \cdot z}(e^{\beta \cdot z})/ \min_{\beta \cdot z}(e^{\beta \cdot z}) $ 
	$\varrho_1 = \min_{\beta \cdot z} (e^{\beta \cdot z})$ and $\varrho_2 = \max_{\beta \cdot z} (e^{\beta \cdot z})$.
	Then, there exists a constant $c$ such that, 
	$$
	\lim_{n\to \infty}  \mathbb{P}  \left(\mathcal{K}(\widehat \theta_{\tau_n},\theta)  
	\leq c \left( \frac{\ln n}{n} \right)^{ 1- \frac{ \varrho_2 + \gamma} { \varrho_1+\gamma+2\alpha\varrho_1}} \right)  = 1,
	%	\leq c \left( \frac{\ln n}{n} \right)^{\frac{\max(e^{\beta \cdot z}) 2\alpha}{\min(e^{\beta \cdot z}) (1+\gamma+2\alpha)}} \right)  = 1,
	$$
	where
	$$
	\tau_n = n^{\frac{  \theta}{\varrho_1+\gamma+2\alpha\varrho_1}} \ln^{-\frac{\theta}{\varrho_1+\gamma +2\alpha\varrho_1}} n.
	$$
\end{theorem}
%Theorem \ref{Th6} can be proved in the same way as Theorem \ref{Th5}. 
The proof of Theorem \ref{Th6} is similar to that of Theorem \ref{Th5} and therefore will not be given in this paper.

\section{The selection of the threshold}\label{autsel}

It is well-known that the choice of the threshold $\tau$ has a major impact on the quality of the estimation in the extreme value modelling. 
We propose a data-driven choice of the threshold $\tau$ inspired by \cite{GramaSpokoiny2008}. The adaptive threshold $\widehat \tau$ is selected by a sequential testing procedure followed by a selection using a penalized maximum likelihood. 

Consider the following semi-parametric survival function consisting of three
parts: 
\begin{equation}
\label{cpmod}
S_{0,\tau,\theta,s,\lambda}(x) = \left\{
\begin{array}{cc}
S_{0}\left( x\right) & \text{if\ \ }x\in \lbrack 0,s ], \\
S_{0}\left( s \right)  \left( \frac{x}{s} \right)^{-\frac{1}{\lambda}} & \text{if\ \ }x \in (s,\tau]. \\
S_{0}\left( s \right)\left( \frac{\tau}{s} \right)^{-\frac{1}{\lambda}} \left( \frac{x}{\tau} \right)^{-\frac{1}{\theta}}& \text{if\ \ }x>\tau ,%
\end{array}%
\right.
\end{equation}
where $\lambda >0$, $\theta >0$ and $x_0<s<\tau$. 
The maximum quasi-likelihood estimators 
$\widehat \theta_\tau$ of $\theta$ and $\widehat \lambda_{s,\tau}$ of $\lambda$ %for fixed $z \in \mathcal{Z}$, 
are respectively given by (\ref{estTheta}) and
\[
\widehat \lambda_{s,\tau}= \frac{\widehat\theta_s\widehat n_s-\widehat\theta_\tau\widehat n_\tau}{\widehat n_{s,\tau}},
\]
where for brevity, we have denoted $\widehat n_s = \sum_{t_i>s}\delta_i$, $\widehat n_\tau = \sum_{t_i>\tau}\delta_i$ and $\widehat n_{s,\tau} = \sum_{s<t_i<\tau}\delta_i$. 

Assume that the observations $t_i$ are ordered in the decreasing order  such that $t_1>...>t_n$. 
We define a uniform grid $K$ in the subscripts $i=1,..,n$ of a size $n_{grid}$, 
%for simplicity, we denote this grid by 
say $K = \{k_1, k_2, ..., k_{n_{grid}} \}$.
The grid $K$ define the set of observations 
$\{t_{k_1}, t_{k_2}, ..., t_{k_{n_{grid}}} \}$
on which the testing procedure 
will be performed. 

We start with the first subscript $k=k_1$ on the grid $K$.
For the subscript $k$ we test the null hypothesis 
$$
\mathcal{H}_0(t_{k}) \quad \text{:} \quad S_0(x)=S_{0,t_{k},\theta}(x)
$$
against the alternative hypothesis
$$
\mathcal{H}_1(t_{k},t_l) \quad \text{:} \quad S_0(x)=S_{0,t_{k},\theta,t_l,\lambda}(x),
$$
where $S_{0,t_{k},\theta}(x)$ is given by \eqref{sem-mod}, $S_0(x)=S_{0,t_{k},\theta,t_l,\lambda}(x)$ 
is given by \eqref{cpmod}, and  $t_l$ can change between $t_1$ and $t_{k}$. 
If we choose all the $t_l$ between $t_1$ and $t_{k}$, some bias is introduced by the observations 
too close to $t_1$ and $t_k$. 
To overcome this problem, we introduce two parameters $\zeta'$ and $\zeta''$ satisfying 
$0<\zeta',\zeta''<0.5$ which are empirically calibrated. 
Now $t_l$ will be varying between $t_{(1-\zeta'')k}$ and $t_{\zeta'k}$. 

The log-likelihood ratio test statistic used to test the null hypothesis $\mathcal{H}_0(t_{k})$ against the alternative $\mathcal{H}_1(t_{k},t_{l})$ is given by
\begin{equation}
\label{LR}
LR(t_{k},t_l)=\widehat n_{t_k,t_l} \mathcal{K}(\widehat\lambda_{t_k,t_l},\widehat\theta_{t_k}) 
+ \widehat n_{t_l} \mathcal{K}(\widehat\theta_{t_l},\widehat\theta_{t_k}),
\end{equation}
where $\mathcal{K}$ is the Kullback-Leibler divergence defined in Section \ref{ConsTheta}. The test statistic is compared to a critical value $D$, 
which is also empirically calibrated. 
We test, for every $t_l \in [t_{(1-\zeta'')k};t_{\zeta'k}]$,
the hypothesis
$\mathcal{H}_0(t_{k})$ against the alternative $\mathcal{H}_1(t_{k},t_l)$
and, if the critical value is not exceeded, we increase the subscript on the grid $K$ and 
preform the test with the next subscript on the grid $K.$
This will be repeated until the critical value is exceeded.
%Consider a positive constant $D$, which will be the critical value in the testing procedure described below. Let be the starting index $k_0\geq3$ and the incremental step for $k$, $k_{step}$. Let $\delta'$ and $\delta''$ be two constants such that $0<\delta',\delta''<0.5$. The parameters of the procedure $k_0$, $k_{step}$, $\delta'$, $\delta''$ and $D$ are to be calibrated empirically.
%We consider that the statistics $T_i$ are ordered in the decreasing order  such that $T_1>...>T_n$. The testing procedure we propose is performed in two steps. In the first step, called the \textit{propagation step}, we test consecutively the null hypothesis $\mathcal{H}_0(T_k)$ : $S_0(\given{x}{z})=S_0(\given{x}{T_k,\theta})$, where $S_0(\given{x}{T_k,\theta})$ is described in (\ref{sem-mod}), against the alternative $\mathcal{H}_1(T_k,T_l)$ : $S_0(\given{x}{z})=S_0(\given{x}{T_k,\theta,T_l,\mu})$, where $S_0(\given{x}{T_k,\theta,T_l,\mu})$ is described in (\ref{cpmod}), for some $\delta'k\leq l \leq (1-\delta'')k$, sequentially in $k = k_0 + ik_{step}$, $i=0,...,[n/k_{step}]$, until the null hypothesis is rejected, where $k_{step}$ is an increment for k to be chosen. 
%The log-likelihood ratio test statistic used to test the null hypothesis $\mathcal{H}_0(T_k)$ against the alternative $\mathcal{H}_1(T_k)$ is given by
%\begin{equation}
%\label{LR}
%LR(T_k,T_l)=\widehat n_{s,\tau} \mathcal{K}(\widehat\mu_{s,\tau},\widehat\theta_s) + \widehat n_\tau %\mathcal{K}(\widehat\theta_\tau,\widehat\theta_s).
%\end{equation}

We denote by $\hat k$ the first subscript $k\in K$ for which the critical value $D$ is exceeded. 
Set  $\widehat s= t_{\hat k}$, which is called in the sequel the \textit{breaking point}. 
%for the first time when the null hypothesis is rejected for a testing threshold $t_{k_j}$. 
We aim to choose the adaptive threshold $\widehat \tau$ by maximizing the quasi-log-likelihood function 
\begin{align*} %\label{}
&\underset{\theta}{\text{max}}\mathcal{L}^{part}(\given{\theta}{\mathbf{z}})  - \text{Pen}(\given{\widehat\theta_{\hat s}}{\mathbf{z}}) \\
&\qquad = \mathcal{L}^{part}(\given{\widehat\theta_\tau}{\mathbf{z}})  - \text{Pen}(\given{\widehat\theta_{\hat s}}{\mathbf{z}}) ,
\end{align*}
where $\text{Pen}(\given{\widehat\theta_{\hat s}}{\mathbf{z}})$ is the penalty function defined by
\begin{equation}
\label{pen}
\text{Pen}(\given{\theta}{\mathbf{z}}) = \mathcal{L}^{part}(\given{\theta}{\mathbf{z}}).
\end{equation}
Taking into account (\ref{pen}), it follows that the second term of (\ref{LR}) can be viewed as the penalized quasi-log-likelihood
\begin{equation}\label{penLik}
\mathcal{L}^{Pen}(s,\tau)= \mathcal{L}^{part}(\given{\widehat\theta_\tau}{\mathbf{z}})-\text{Pen}(\given{\widehat\theta_{s}}{\mathbf{z}}) .
\end{equation}
We find the subscript which maximize the penalized quasi-log-likelihood
\begin{align} \label{choice-lll001}
\widehat l = \underset{\zeta'\hat k\leq l \leq (1-\zeta'')\hat k}{\mbox{argmax} } 
\mathcal{L}^{Pen}\left( \hat s , t_l \right).
\end{align}
Finally, we set the adaptive threshold 
\begin{align} \label{choice-tau001}
\widehat \tau = t_{\hat l}
\end{align} 
and its associated parameter $\widehat \theta_{\hat \tau}$.

%Let be the critical value $D>0$ in the procedure presented below. The values of the parameters $\delta_0$, $\delta'$, $\delta''$ and $k_{step}$ are calibrated by simulations. The algorithm used to choose the adaptive threshold $\widehat \tau$ is as follows:
%\begin{itemize}
%\item \textbf{Step 1.} Set the starting index $k = k_0$
%\item \textbf{Step 2.} Compute the test statistic for testing $\mathcal{H}_{T_k}(z)$ against $\tilde{\mathcal{H}}_{T_k}(z)$ :
%\[
%LR_{max}(T_k) = \underset{\delta' k\leq l \leq (1-\delta'') k}{max} LR(T_k,T_l).
%\]
%\item \textbf{Step 3.} If $k \leq n- k_{step}$ and $LR_{max}(T_k) \leq D$, increase $k$ by $k_{step}$ and go to Step 2. If $k > n- k_{step}$ or $LR_{max}(T_k) > D$, let $\widehat k = k $ and define
%\[
%\widehat l = \underset{\delta'\widehat k\leq l \leq (1-\delta'')\widehat k}{argmax } \mathcal{L}^{Pen}\left( T_{\widehat k}, T_l \right).\]
%
%\end{itemize}

\section{Aggregation}\label{Aggreg}

The transition between the non-parametric part and the parametric part in the model \eqref{sem-mod} can sometimes be rough,
especially when the sample size is small.  
We propose two ways of smoothing the transition relying on a aggregating estimators corresponding to different thresholds.

\subsection{Simple aggregation}\label{SimpleAggreg}
The first aggregation we describe can be called 'simple aggregation' as we aggregate the estimated cumulative hazard function from multiple thresholds. The procedure can be resumed with 3 simple steps, where the observations $t_i$ are ordered in the decreasing order  such that $t_1>...>t_n$.
\begin{itemize}
	\item \textbf{Step 1.} Choose $M\geq1$ thresholds $\tau_1=t_{m_0},\dots, \tau_m=t_{m_0+M-1}$ from the observed values, where $m_0\geq 1.$  
	\item \textbf{Step 2.} For each chosen threshold $\tau_k$ compute the estimated cumulative hazard function $\widehat H_{z,\tau_k,\theta} (x)$.
	\item \textbf{Step 3.} Compute the simple aggregation estimator by % Aggregate the estimated cumulative hazard functions by the simple formula
\begin{equation}\label{sae}
	\widehat S_{\mbox{\scriptsize sa}}(\given{x}{z}) = \exp \left( - \frac{1}{M}\sum_{k=1}^{M} \widehat H_{z,\tau_k,\theta} (x) \right).
\end{equation}
\end{itemize}
%the threshold among the last observations and at least three observed values after the last chosen threshold. 
For the algorithm to work, we need to choose $m_0$ as the first observation (censored or not) having at least one non-censored observation above it:
$m_0 \geq \min \{ m: \sum_{t_i > t_m}\delta_i =1\}.$ 
With $M=1$ and $m_0=k$, the procedure becomes the estimation of the semi-parametric model \eqref{sem-mod} with a fixed threshold $\tau=t_{k}$.

\subsection{Adaptive aggregation}\label{AddAggreg}

The second aggregation we describe can be called 'adaptive aggregation' as we aggregate cumulative hazard functions from the adaptive procedure described in Section \ref{autsel}. Let $l_1,\dots,l_{N_{\hat s}}$ be the reversed ranks of the sequence
$$ 
\mathcal{L}^{Pen}\left( t_{\hat s}, t_l \right), \ \zeta'\hat s\leq l \leq (1-\zeta'')\hat s,
$$
where $N_{\hat s}$ is its cardinality
and $\hat s$ is the breaking point computed by the adaptive procedure described in Section \ref{autsel}. 
%\leq \dots \leq l_{min} =  \underset{\zeta'\hat s\leq l \leq (1-
%\zeta'')\hat s}{argmin } \mathcal{L}^{Pen}\left( t_{\hat s}, t_l \right).$$
%$$l_1 =  \underset{\zeta'\hat s\leq l \leq (1-\zeta'')\hat s}{argmax } \mathcal{L}^{Pen}\left( t_{\hat s}, t_l \right)\leq \dots \leq l_{min} =  \underset{\zeta'\hat s\leq l \leq (1-%\zeta'')\hat s}{argmin } \mathcal{L}^{Pen}\left( t_{\hat s}, t_l \right).$$
For the adaptive aggregation we proceed in the same way as in the case of the simple aggregation described above.  
It can be resumed with the following steps: 
\begin{itemize}
	\item \textbf{Step 1.} 
	%From the results of the adaptive procedure explained in Section \ref{autsel}, 
	Choose $M\geq1$ thresholds $\tau_1=t_{l_1},\dots, \tau_M=t_{l_M}$ from the observed values.
	%are chosen from the thresholds maximizing the penalized likelihood \eqref{penLik}, $\tau_{l_1}, ..., \tau_{l_j}$.
	\item \textbf{Step 2.} For each chosen threshold $\tau_{l_k}$, compute the estimated cumulative hazard function $\widehat H_{z,\tau_{l_k},\theta} (x)$.
	\item \textbf{Step 3.} Compute the weights on the estimated cumulative hazard functions from the value of the penalized likelihood \eqref{penLik} by
	$$
	w_{l_k} = \mathcal{L}^{Pen}\left( t_{\hat s}, t_{l_k} \right) / \sum_{i=1}^{M} \mathcal{L}^{Pen}\left( t_{\hat s}, t_{l_i} \right).
	$$
	\item \textbf{Step 4.} Compute the adaptive aggregation estimator by
	\begin{equation}\label{Aae}
	\widehat S_{\mbox{\scriptsize aa}}(\given{x}{z}) = \exp \left( - \sum_{k=1}^{M} w_{l_k} \widehat H_{z,\tau_{l_k},\theta} (x) \right).
	\end{equation}
\end{itemize}
Note that with $M=1$, the procedure becomes the estimation of the semi-parametric model \eqref{sem-mod} with the adaptive threshold $\tau$ chosen from the procedure described in Section \ref{autsel}.

\section{Simulation study}\label{sim}

We carry a simulation study to evaluate how 
the proposed estimators 
%based on simple aggregation 
%on the adaptive procedure and on the adaptive aggregation 
behave 
against the usual Nelson-Aalen estimator. 
We are interested in the values of the baseline survival probability $S_0(x)$ when $x$ is large. To compare the estimations, we use the relative mean square error (RelMSE), which we define as $RelMSE_{\widehat{S_0}(x)} = \mathbb{E} \left(\ln^2\frac{\widehat S_0 (x)}{S_0 (x)}\right)$. One can compare the estimated survival function and the true survival function for any $z$ by multiplying the error for the baseline by $e^{2\beta . z}$. We also look at the ratio between the estimators proposed in this paper and the Nelson-Aalen estimator. The parameters of the adaptive procedure in Section \ref{autsel} are set to the following values $n_{\text{grid}}=100$, $\zeta'=0.25$, $\zeta''=0.05$.
A simulation study has been performed in \cite{DGPT2015} concerning the choice of these parameters and led to these values.

The study of the properties of the adaptive estimators introduced in Sections \ref{autsel} and \ref{Aggreg} 
is outside the scope of this paper.   
However, we shall discus briefly the choice of the critical value $D$, as this seems to be one of the sensitive points of the proposed adaptive procedure in Section \ref{autsel}. 
%We refer for more details to \cite{GramaSpokoiny2008}. 
In our paper the critical value is calibrated under the hypothesis that the true distribution $F$ is fully parametric and follows a Pareto distribution with some parameter $\theta$ and so does not depend on $F_T$ nor $F_C$.
We note in addition that, as it is shown in \cite{GramaSpokoiny2008}, 
the law of the test statistic \eqref{LR} does not depend on the parameter $\theta$, %of the used Pareto law, 
and therefore, we can choose $\theta=1$. It is also argued in \cite{GramaSpokoiny2008} that the obtained critical value remains stable with the change of the number of observations. 
If we suppose for a moment that we introduce censoring in the model, at least from the intuitive point of view, it will act as a reducer of the number of extreme observations, so by the previous remark the critical value will remain stable.
Moreover, it was also observed that the adaptive choice $\widehat \tau$ remains stable with respect to some variations of $D$,
since $\widehat \tau$ is defined by   
maximizing the penalized quasi-log-likelihood, see \eqref{choice-lll001} and \eqref{choice-tau001} which stabilize the value of $\widehat\tau$.

The Figure \ref{TSsim} 
shows the empirical distribution function of the test statistic computed from $10000$ Monte-Carlo simulations
with standard Pareto law ($\theta=1$) as true baseline distribution without censoring (continuous line).
We performed an additional study of the critical value in the same conditions under a Pareto type censoring with parameters 
$\theta =\frac{1}{2}$ and $\theta=2$. 
The behaviour of the critical value under the censoring is given in Figure \ref{TSsim} as dashed and dotted lines, respectively. 
The Figure \ref{TSsim} shows that the censoring does not affect significantly the choice of the critical value $D$.

Let be the transformed Cauchy distribution with location parameter $x_0$ and scale $\gamma$ defined as: 
\begin{equation*}
F(x) = 1 - \frac{1 - \frac{1}{\pi} \arctan{\left(\frac{x-x_0}{\gamma}\right)} + \frac{1}{2}}{1- \frac{1}{\pi} \arctan{\left(\frac{0-x_0}{\gamma}\right)} + \frac{1}{2}}
\end{equation*}

\begin{figure}[h]
	\begin{center} 
		\includegraphics[width=90mm,height=60mm]{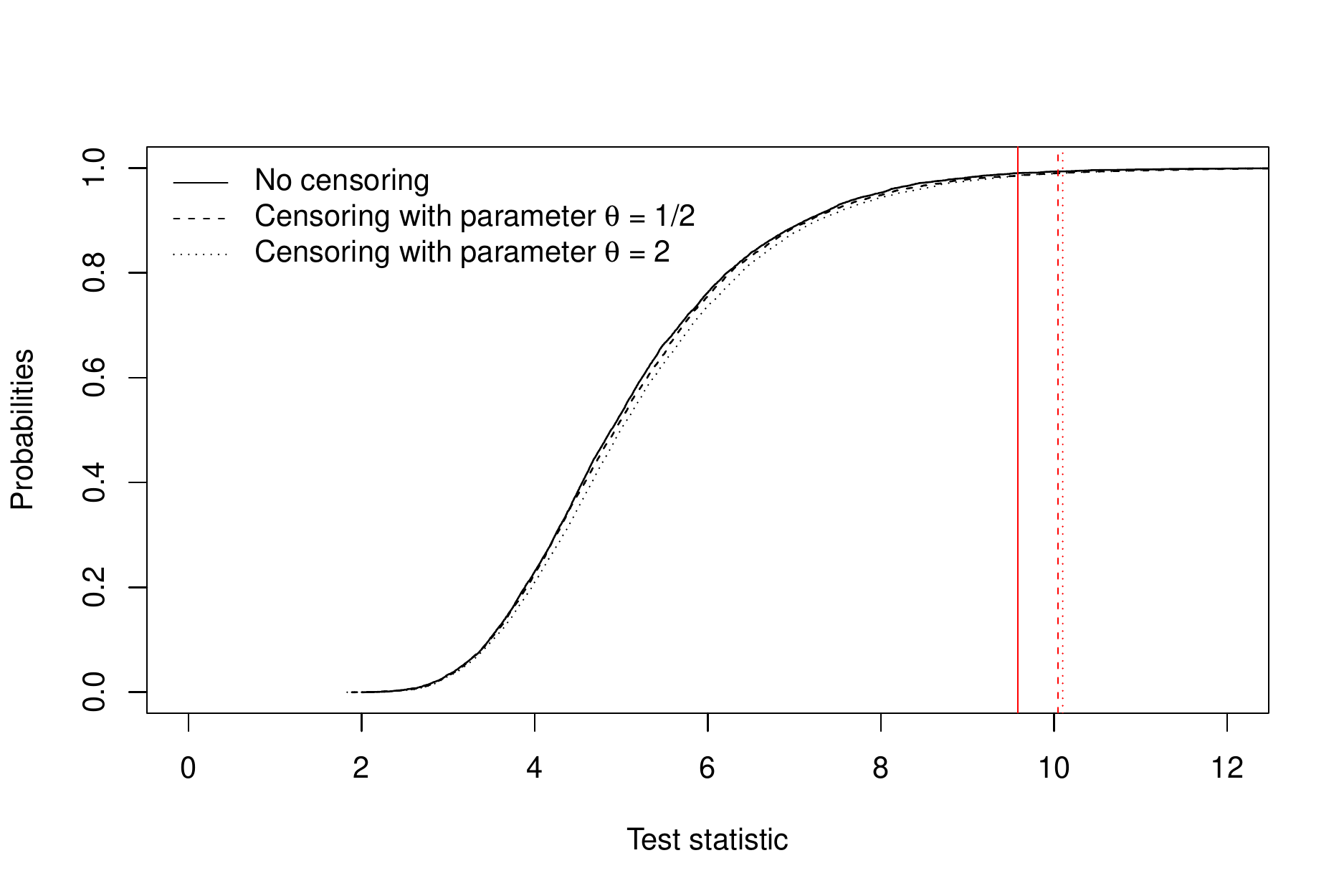} 
		\caption{Empricial distribution function of the test statistic. 
		The three $0.99$-quantiles are marked by the corresponding vertical lines.}
		\label{TSsim}
	\end{center}
\end{figure}

The modified Cauchy model that we use for our simulation is one of the difficult models. It represents a typical model for which only the tail of the 
observed data (beyond the threshold $\tau$) is of type "heavy tailed", while the "beginning" and the "middle" (the part before the threshold $\tau$) 
of the distribution is not. 
In our case the part before the threshold $\tau$ can be of any form and is estimated non-parametrically (by the non-parametric Nelson-Aalen estimator). The tail is detected by our threshold selection procedure formulated in Section \ref{autsel}.

We compare the estimator defined by \eqref{EMSP} with the usual non-parametric Nelson-Aalen  estimator. Moreover, we have also compared both of them with the estimators described in \eqref{SimpleAggreg} and \eqref{AddAggreg}.
The number of aggregation is set by simulations to $M=40$ in the procedure described in Section \ref{SimpleAggreg} and \ref{AddAggreg}.

\begin{figure}[h]
	\begin{center} 
		\includegraphics[width=90mm,height=60mm]{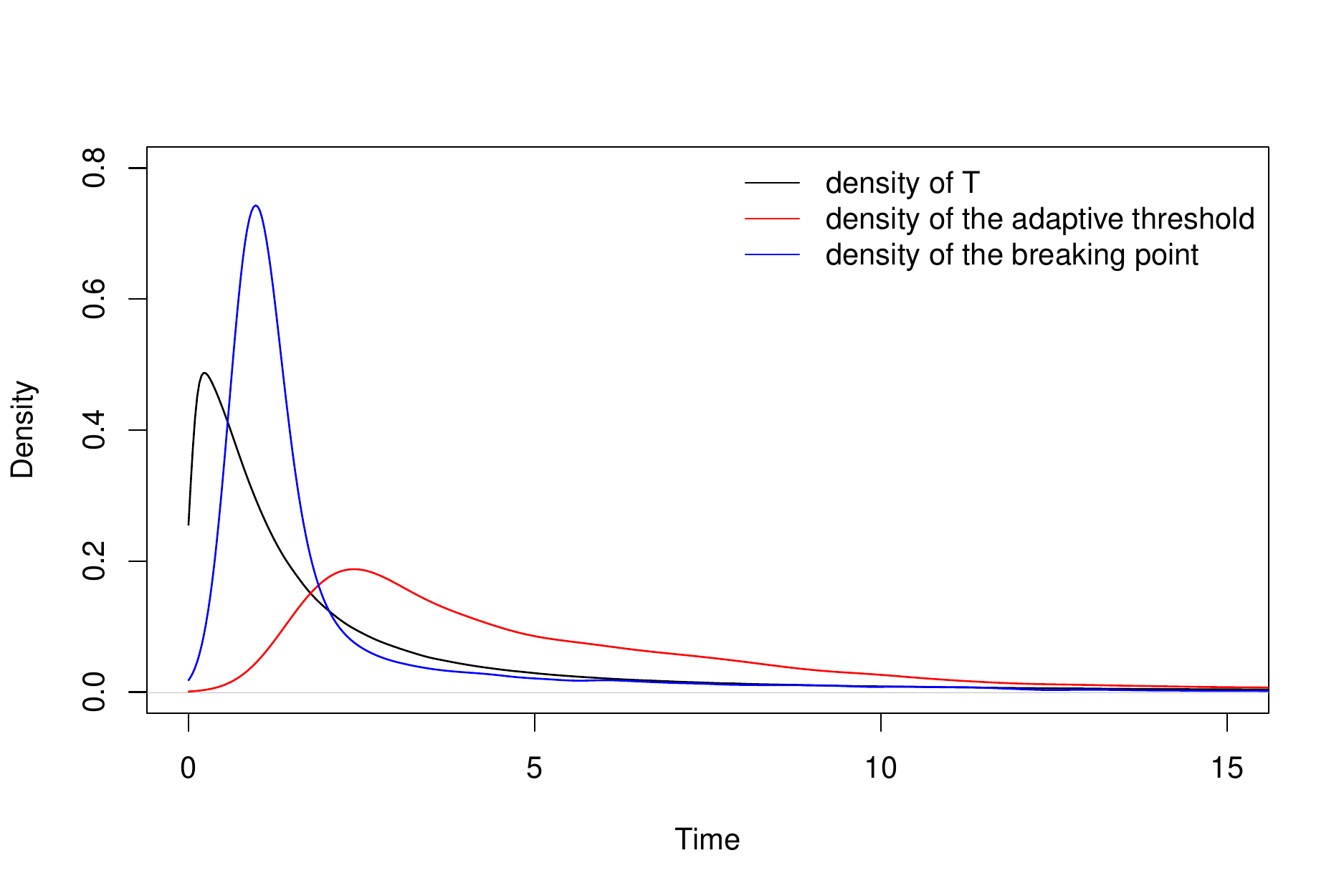} 
		\caption{Density of the threshold $\widehat{\tau}$, the breaking point $\widehat{s}$ and the observations $T$ computed on $10000$ simulations of the transformed Cauchy distribution with parameters $x_0=0$ and $\gamma=1$.}
		\label{DensEx}
	\end{center}
\end{figure}

In order to illustrate the choice of the transformed Cauchy distribution with location $x_0=0$ and scale $\gamma=1$, we show in Figure \ref{DensEx} the density of the failure time $X$, the density of the threshold $\widehat{\tau}$ and the breaking point $\widehat{s}$ chosen by the adaptive procedure presented in Section \ref{autsel}. One can see that the breaking point is chosen after the location parameter of the transformed Cauchy distribution and the threshold is therefore set to a value greater than it.
\begin{table}[h]
	\begin{center} 
		\begin{tabular}{|c|c|c|c|c|c|}
			\hline
			\multicolumn{6}{|c|}{$n=100$}\\
			\hline
			\hline
			$x$ & $100$ & $200$ & $300$ & $400$ & $500$\\
			\hline
			%$S_0(x)$ & $0.5570$ & $0.0503$&$0.0208$&$0.0075$\\
			%Mean of $\widehat{S}_0(x)$ & $0.5586$ & $0.0897$&$0.0870$&$0.0865$\\
			%Mean of $\widehat{S}_0(\given{x}{\tau,\beta,z})$ & $0.5180$ & $0.0302$&$0.0009$&$0.00003$\\
			%Mean of $\widehat{S}_0(\given{x}{\widehat \tau ,\beta,z})$ & $0.5618$ & $0.0309$&$0.0067$&$0.0013$\\
			RelMSE of $\widehat{S}_0(x)$ & 4.6750&  7.8920& 10.2147& 12.0651& 13.6145 \\
			%MARE of $\widehat{S}_0(\given{x}{\tau,\beta,z})$ & $0.0799$ & $0.4324$&$0.9356$&$0.9937$\\
			RelMSE of $\widehat{S}_{0,\widehat \tau ,\beta}(x)$ & 1.5908& 2.1641& 2.5413& 2.8277 &3.0605\\
			RelMSE with simple aggregation &1.0306 &1.4164 &1.6714 &1.8654& 2.0234\\
			RelMSE with adaptive aggregation&1.2361& 1.6976 &2.0023 &2.2341 &2.4229 \\
			\hline
			\hline
			\multicolumn{6}{|c|}{$n=500$}\\
			\hline
			\hline
			$x$ & $100$ & $200$ & $300$ & $400$&$500$\\
			\hline
			RelMSE of $\widehat{S}_0(x)$ &2.0733& 4.1024& 5.7537& 7.1133& 8.2953 \\
			%MARE of $\widehat{S}_0(\given{x}{\tau,\beta,z})$ & $0.0799$ & $0.4324$&$0.9356$&$0.9937$\\
			RelMSE of $\widehat{S}_{0,\widehat \tau ,\beta}(x)$ & 0.4209 &0.5844& 0.6928& 0.7754& 0.8427\\
			RelMSE with simple aggregation & 0.4149 &0.5771 &0.6846& 0.7666 &0.8335\\
			RelMSE with adaptive aggregation& 0.3763& 0.5253& 0.6244& 0.6999& 0.7616\\
			\hline
		\end{tabular}
	\end{center}
	\caption{$1000$ Monte-Carlo simulations where the parameters of the censorship distribution are $x_0=0$ and $\gamma=2$.}
	\label{SimCauch1}
\end{table}
Assume that we have decided on the sample size $n$, the parameter $\beta$, the covariate distribution of $Z$, the baseline survival 
function $S_0(\cdot)$ and the censoring distribution $S_C(\given{\cdot}{z})$. We generated data sets in our simulation study following the pattern : 
we consider $F_0$ to follow the transformed Cauchy distribution with parameters $x_0=0$ and $\gamma=1$. 
In the first study, the survival distribution function $S_C$ is assumed following the Cauchy distribution with parameter $x_0=0$ and $\gamma=2$. 
In a second study, $S_C$ is assumed following the transformed Cauchy distribution with parameter $x_0=10$ and $\gamma=0.1$.
In both cases, the censoring survival distribution function doesn't depend on the covariates.
%The figure \ref{SSC} shows that there is more censorship at the end of the tail than the rest of the distribution of $S_0$. 
The mean censoring rate is around $50\%$ for the first distribution of $S_C$ and around $20\%$ for the second distribution of $S_C$.
\begin{figure}[h]
	\makebox[\textwidth][c]{\begin{tabular}{cc}
			\includegraphics[width=79mm,height=52.5mm]{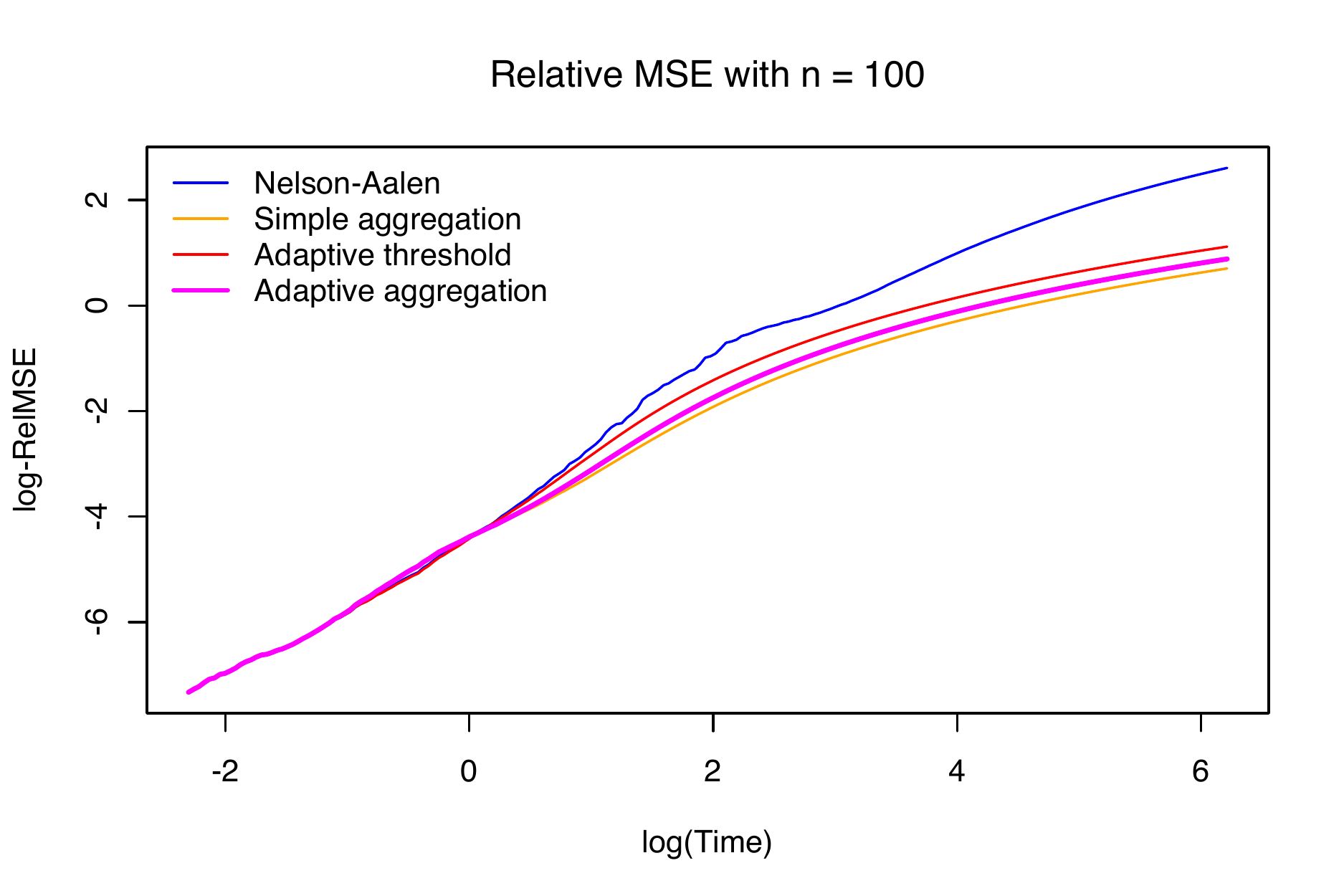}&\includegraphics[width=79mm,height=52.5mm]{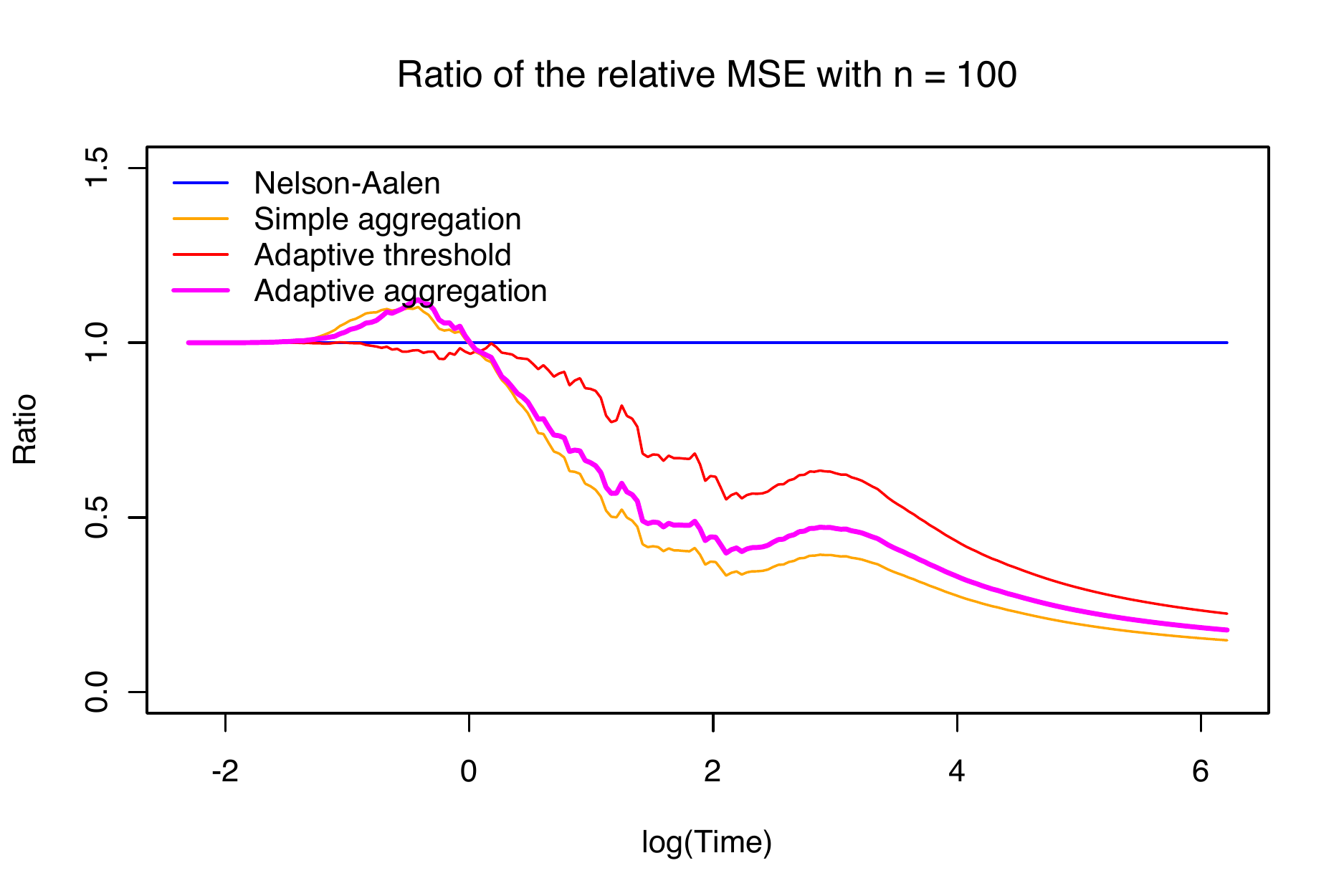}\\
			\includegraphics[width=79mm,height=52.5mm]{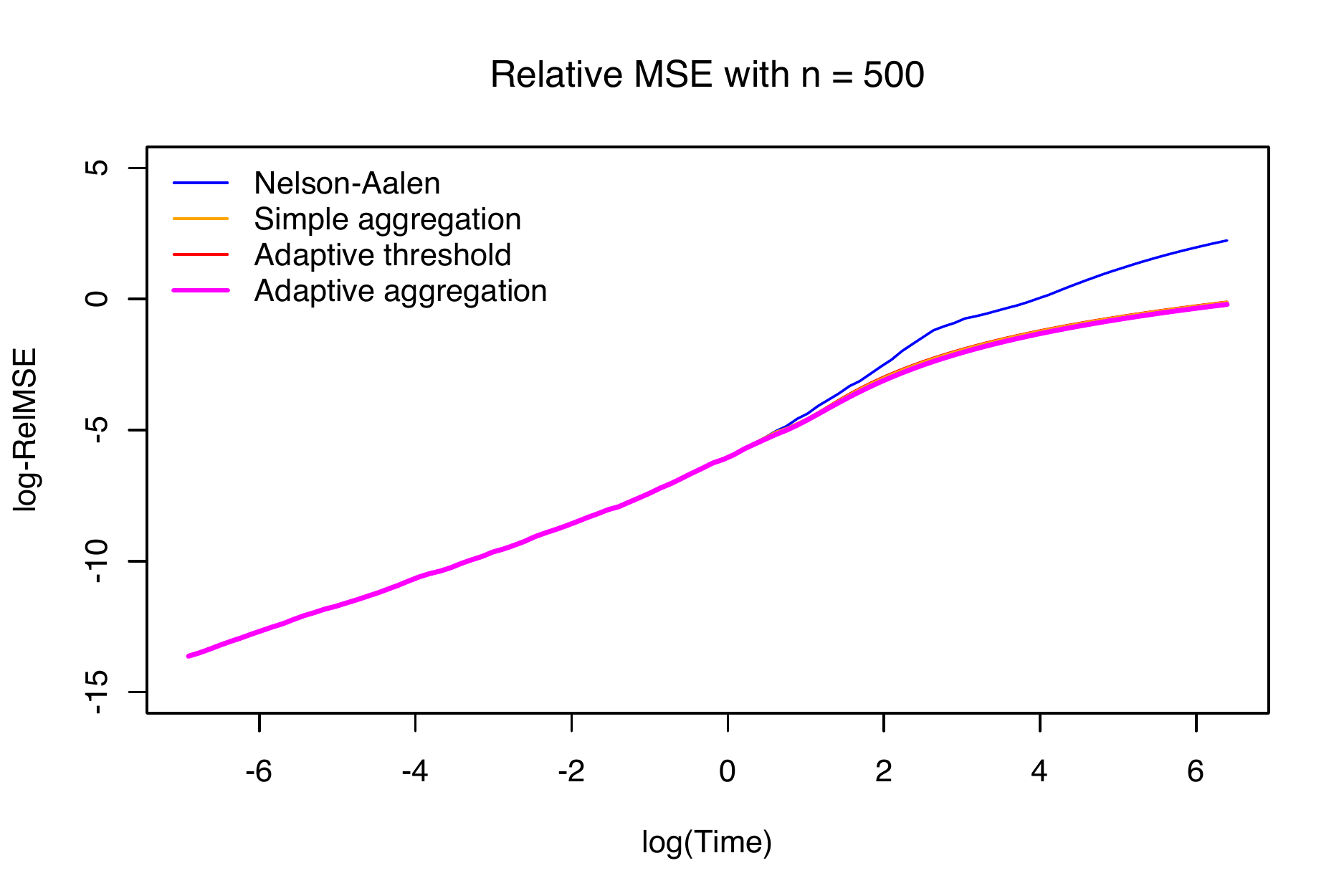}&\includegraphics[width=79mm,height=52.5mm]{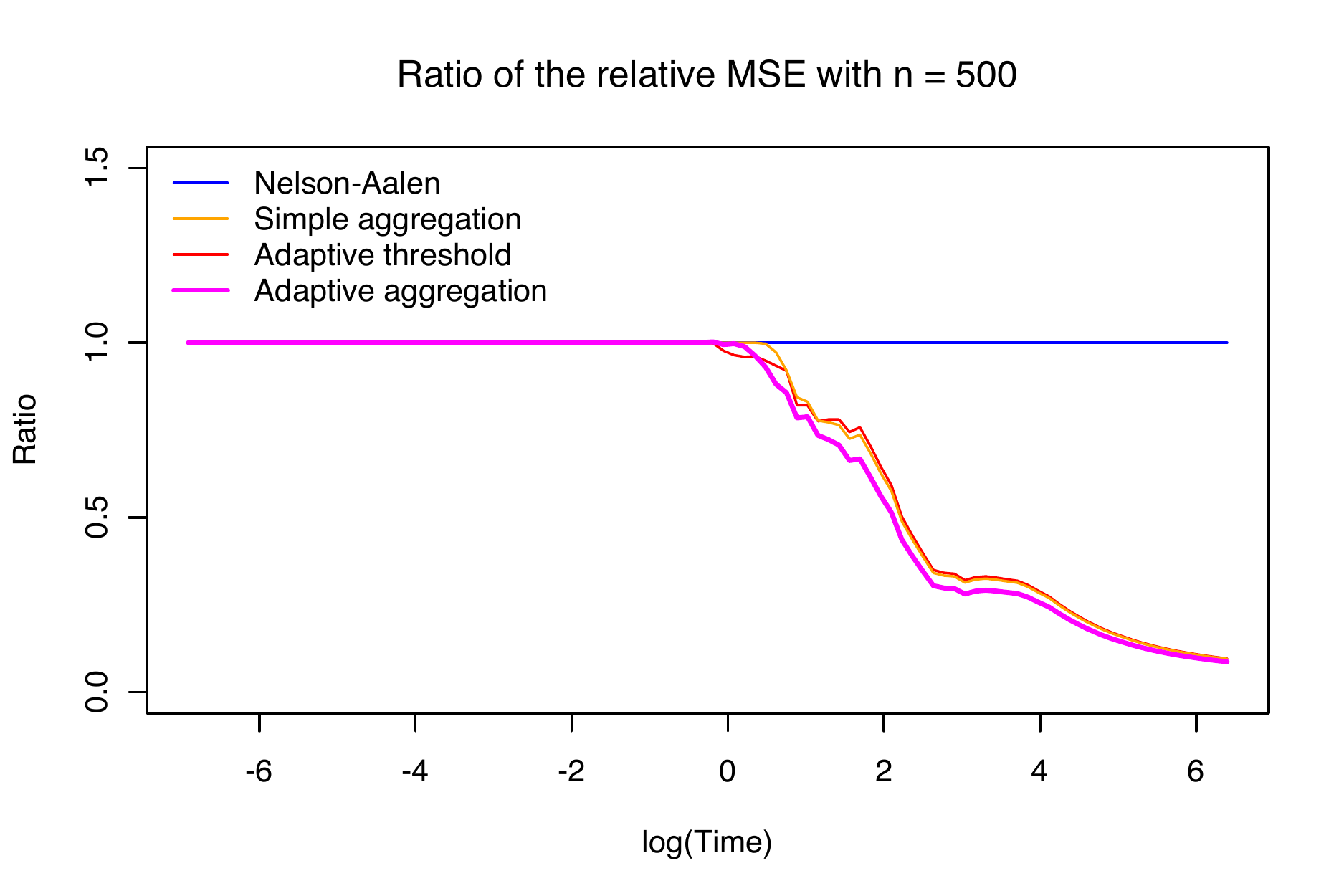}
	\end{tabular}}
	\caption{Relative MSE (left) and ratio of the relative MSE (right) for $n=100$ (top) and $n=500$ (bottom) and $NMC = 1000$. The parameters of the censorship distribution are $x_0=0$ and $\gamma=2$.}
	\label{SimCauchFig}
\end{figure}
The covariate is supposed to be a random uniform variable in $[-1,1]$ and the parameter $\beta$ is set to $-0.5$.
%The parameter $\tau$ is considered following the algorithm in Section \ref{autsel} in the model $\widehat{S}_0(\given{x}{\widehat \tau ,\beta,z})$.
For both $n=100$ and $n=500$, we chose $M=10$ thresholds for the aggregation procedure (simple and adaptive). 
For the simple aggregation, we chose $m_0$ 
corresponding to approximately
%between %the last $20\%$ (for $m_0$) and 
$6\%$ of the observed values.

Table \ref{SimCauch1} gives the results of a Monte-Carlo simulation for the estimation of the baseline function with $S_C$ following a transformed Cauchy distribution with parameters $x_0=0$ and $\gamma=2$.
\begin{table}[H]
	\begin{center} 
		\begin{tabular}{|c|c|c|c|c|c|}
			\hline
			\multicolumn{6}{|c|}{$n=100$}\\
			\hline
			\hline
			$x$ & $100$ & $200$ & $300$ & $400$ & $500$\\
			\hline
			%$S_0(x)$ & $0.5570$ & $0.0503$&$0.0208$&$0.0075$\\
			%Mean of $\widehat{S}_0(x)$ & $0.5586$ & $0.0897$&$0.0870$&$0.0865$\\
			%Mean of $\widehat{S}_0(\given{x}{\tau,\beta,z})$ & $0.5180$ & $0.0302$&$0.0009$&$0.00003$\\
			%Mean of $\widehat{S}_0(\given{x}{\widehat \tau ,\beta,z})$ & $0.5618$ & $0.0309$&$0.0067$&$0.0013$\\
			RelMSE of $\widehat{S}_0(x)$ &5.1196 & 8.6637 &11.1824 &13.1688& 14.8236\\
			%MARE of $\widehat{S}_0(\given{x}{\tau,\beta,z})$ & $0.0799$ & $0.4324$&$0.9356$&$0.9937$\\
			RelMSE of $\widehat{S}_{0,\widehat \tau ,\beta}(x)$ &  0.9857& 1.3812 &1.6447 &1.8460& 2.0104\\
			RelMSE with simple aggregation&0.7266 &1.0534 &1.2749 &1.4456 &1.5857\\
			RelMSE with adaptive aggregation&0.5970 &0.8203& 0.9681 &1.0807 &1.1725 \\
			\hline
			\hline
			\multicolumn{6}{|c|}{$n=500$}\\
			\hline
			\hline
			$x$ & $100$ & $200$ & $300$ & $400$&$500$\\
			\hline
			RelMSE of $\widehat{S}_0(x)$ & 4.4576&  7.7855& 10.1777 &12.0743& 13.6595\\
			%MARE of $\widehat{S}_0(\given{x}{\tau,\beta,z})$ & $0.0799$ & $0.4324$&$0.9356$&$0.9937$\\
			RelMSE of $\widehat{S}_{0,\widehat \tau ,\beta}(x)$ & 0.2162& 0.3128& 0.3780& 0.4281& 0.4691\\
			RelMSE with simple aggregation& 0.4499 &0.7033 &0.8787& 1.0154& 1.1283\\
			RelMSE with adaptive aggregation&0.1597& 0.2257 &0.2698 &0.3036& 0.3313\\
			\hline
		\end{tabular}
	\end{center}
	\caption{$1000$ Monte-Carlo simulations with the following parameters of the censorship distribution $x_0=10$ and $\gamma=0.1$.}
	\label{SimCauch2}
\end{table}
\begin{figure}[h]
	\begin{center} 
		\includegraphics[width=90mm,height=60mm]{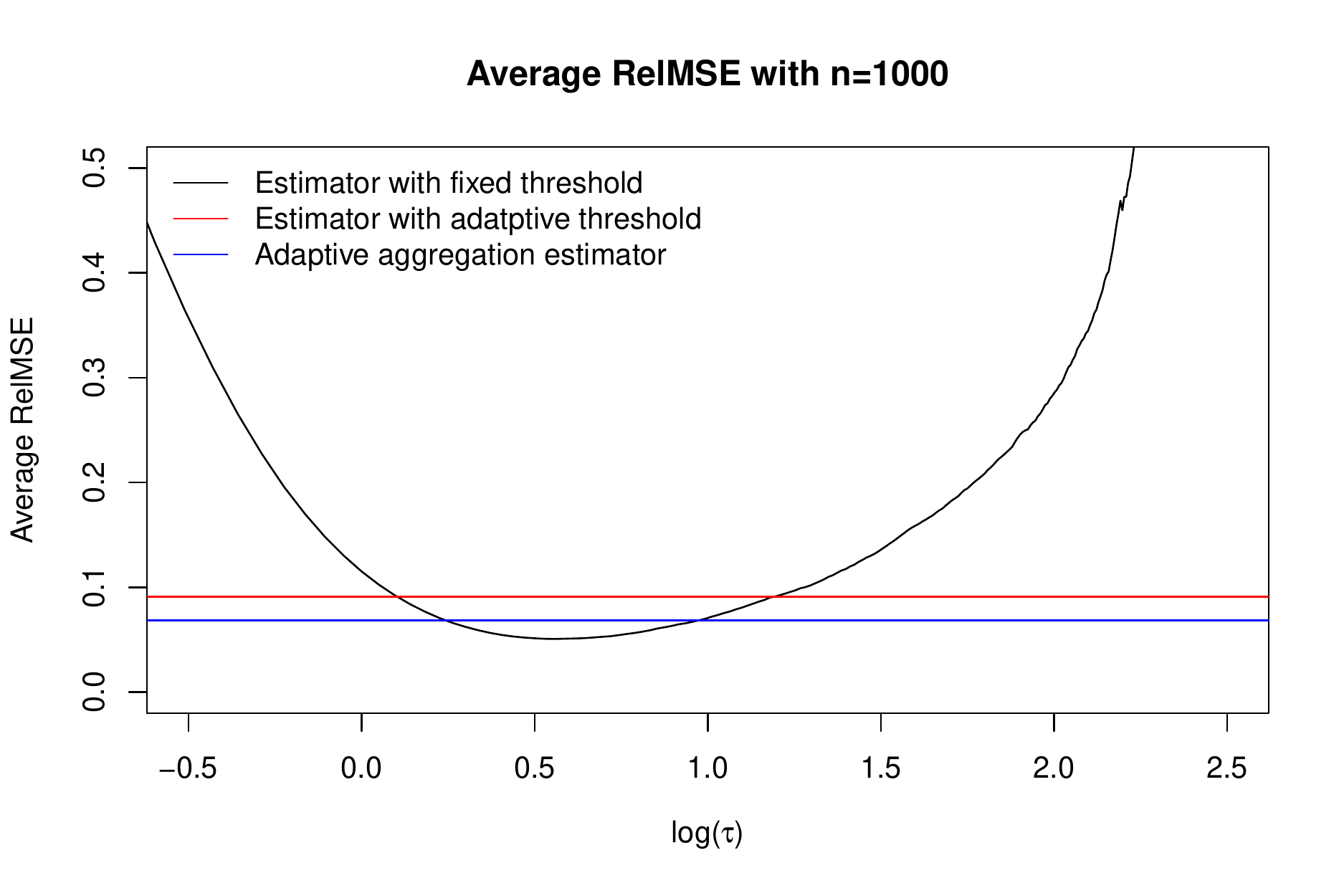} 
		\caption{Average RelMSE's of the estimator with fixed threshold $\tau$ (black line), of the estimator with the adaptive choice of the threshold (red line) and of the adaptively aggregated estimator (blue line). All estimators are computed on $10000$ 
			simulations of the transformed Cauchy distribution with parameters $x_0=0$ and $\gamma=1$ and with parameters $x_0=0$ and $\gamma=2$ for the censorship .}
		\label{TauChoice}
	\end{center}
\end{figure}
Table \ref{SimCauch2} gives the results of a Monte-Carlo simulation for the estimation of the baseline function with $S_C$ following a transformed Cauchy distribution with parameters $x_0=10$ and $\gamma=0.1$.
The performance of the adaptive choice of the threshold is shown 
%performed a simulation study on the choice of $\tau$. 
%%%%%%%%%%%%%
%\todo{Changer la couleur "Green" en "Blue"}
%%%%%%%%%%%%
on Figure \ref{TauChoice}. We plot the average RelMSE of the estimated survival function \eqref{EMSP} with fixed threshold 
$\tau$, where $\tau$ is 
running on the uniform grid on $[0.1, 20]$ (black line) and the average
RelMSE of the estimated survival function \eqref{EMSP} with the adaptive threshold $\widehat\tau$ (red line) chosen by the procedure described in Section \ref{autsel}. 
The average RelMSE is defined by : $ARelMSE_{\widehat{S_0}} = 1/n_{grid} \sum_{j=1}^{n_{grid}} RelMSE_{\widehat{S_0}(x_j)}$, where $x_j, j=1,...,n_{grid}$ is a geometric grid on $[0.1,100]$.
The average RelMSE of the aggregated estimated survival function $\widehat S_{0,\mbox{\scriptsize aa}}$ is also shown in blue line.
% : i.e. $ARelMSE = 1/n_{grid} \sum_{j=1}^{n_{grid}}\widehat{H}_0(x_j)$, where $x_1,...,x_{n_{grid}}$ is a geometric grid in $[0,100]$.  
One can see that the adaptive choice of the threshold is not optimal but it is close to the best one. Recall that the choice of the adaptive threshold is done without any prior knowledge of the 'best' choice of the threshold (which sometimes is also called 'oracle' choice).

We have added another model based on the $\log$-Gamma law, which has a rather "bad" slowly varying part of type $\log x$. The results are reported in Figure \ref{SimLogGammaFig} and Table \ref{SimLogGamma}. These results show that the introduction of the aggregated estimator improves the estimation of the tails significantly with respect to standard Nelson-Aalen estimator. In particular, for low sample sizes, the aggregation improves the error estimation over the simple adaptive choice. The simulated distribution of the survival and censoring time follow a log-gamma distribution where the density function is defined for $x>1$ by:
$$
f(x) = \frac{b^a \ln(x)^{a - 1}}{\Gamma(a)x^{b+1}},
$$
where $a>0$ is the shape parameter, $b >0$ is the rate parameter and $\Gamma(x) = \int_0^\infty t^{x-1} \exp(-t) dt$ is the Gamma function. We considered $F_0$ to follow a log-Gamma distribution with parameters $a=2$ and $b=2$ and $S_C$ to follow a log-Gamma distribution with parameters $a=5$ and $b=3.5$. The mean censoring rate is around $60\%$.

\begin{figure}[h]
	\makebox[\textwidth][c]{\begin{tabular}{cc}
			\includegraphics[width=79mm,height=52.5mm]{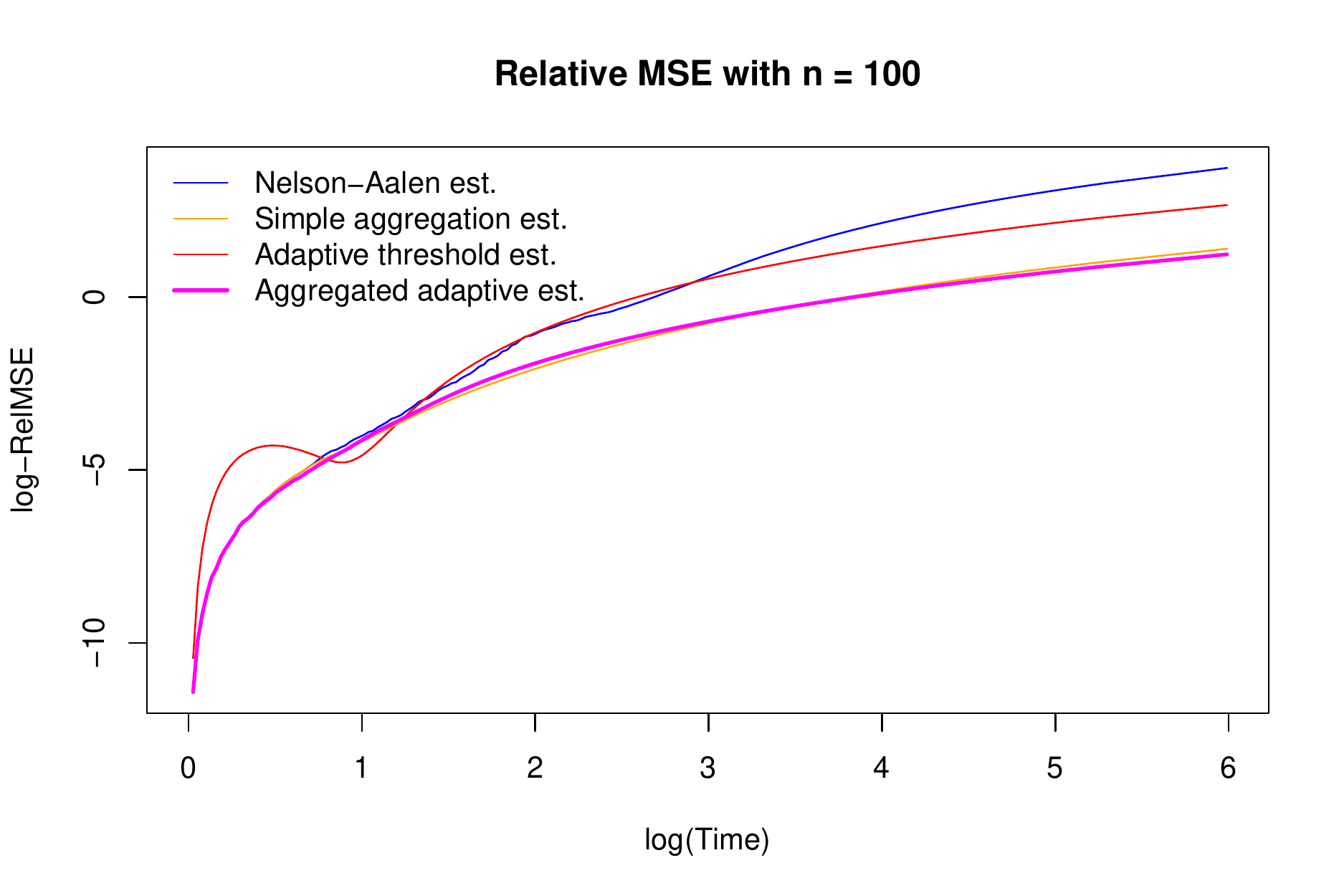}&\includegraphics[width=79mm,height=52.5mm]{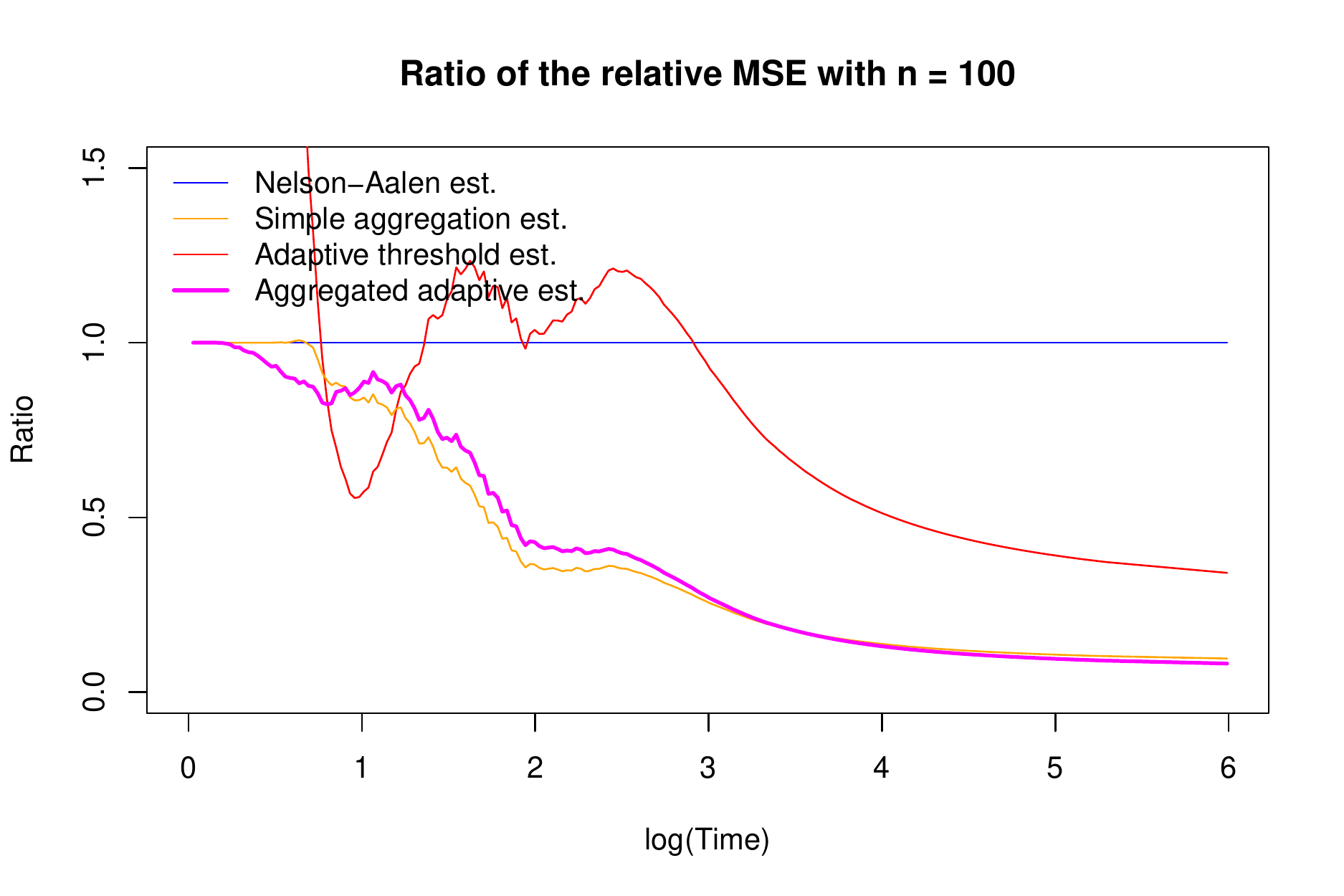}\\
			\includegraphics[width=79mm,height=52.5mm]{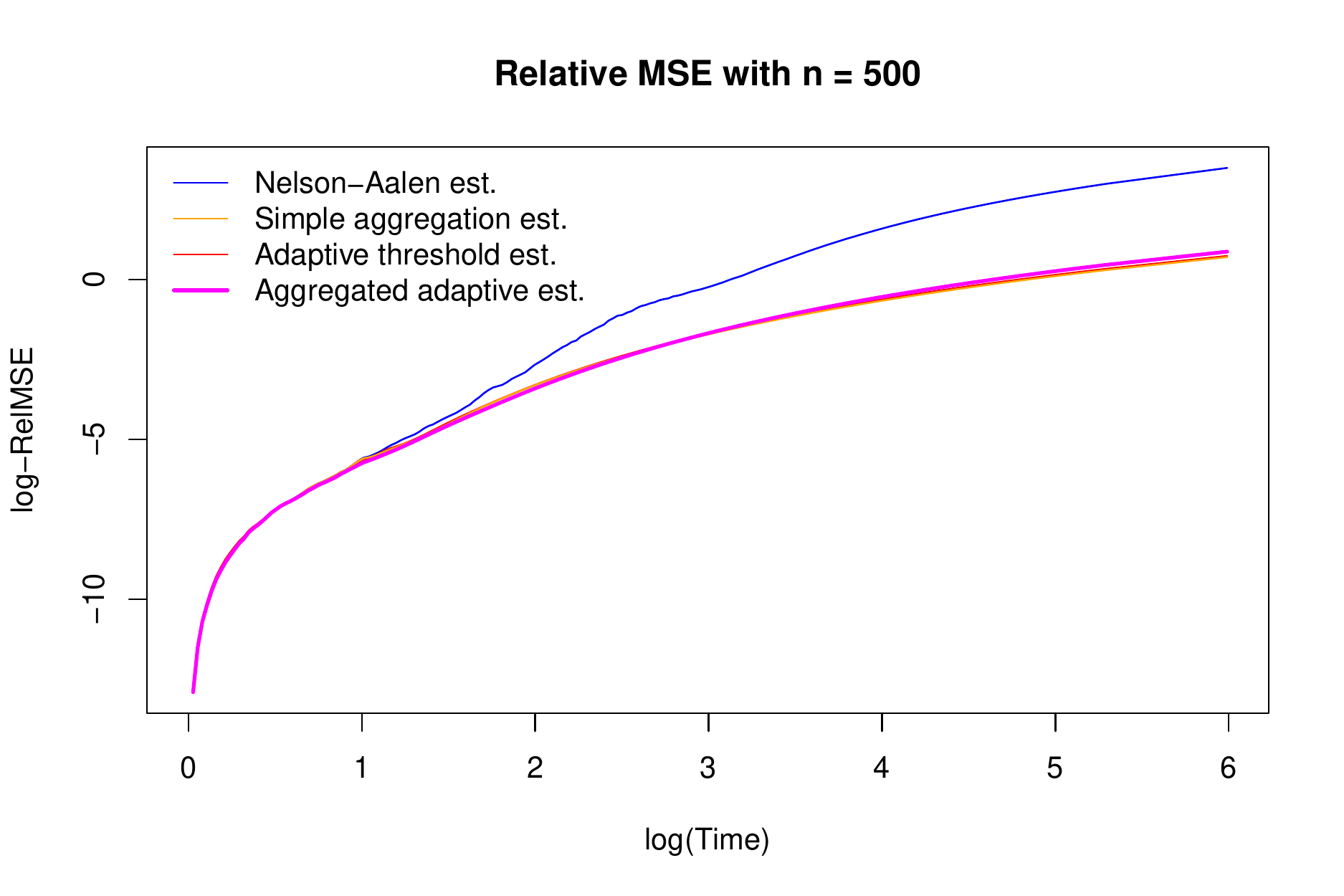}&\includegraphics[width=79mm,height=52.5mm]{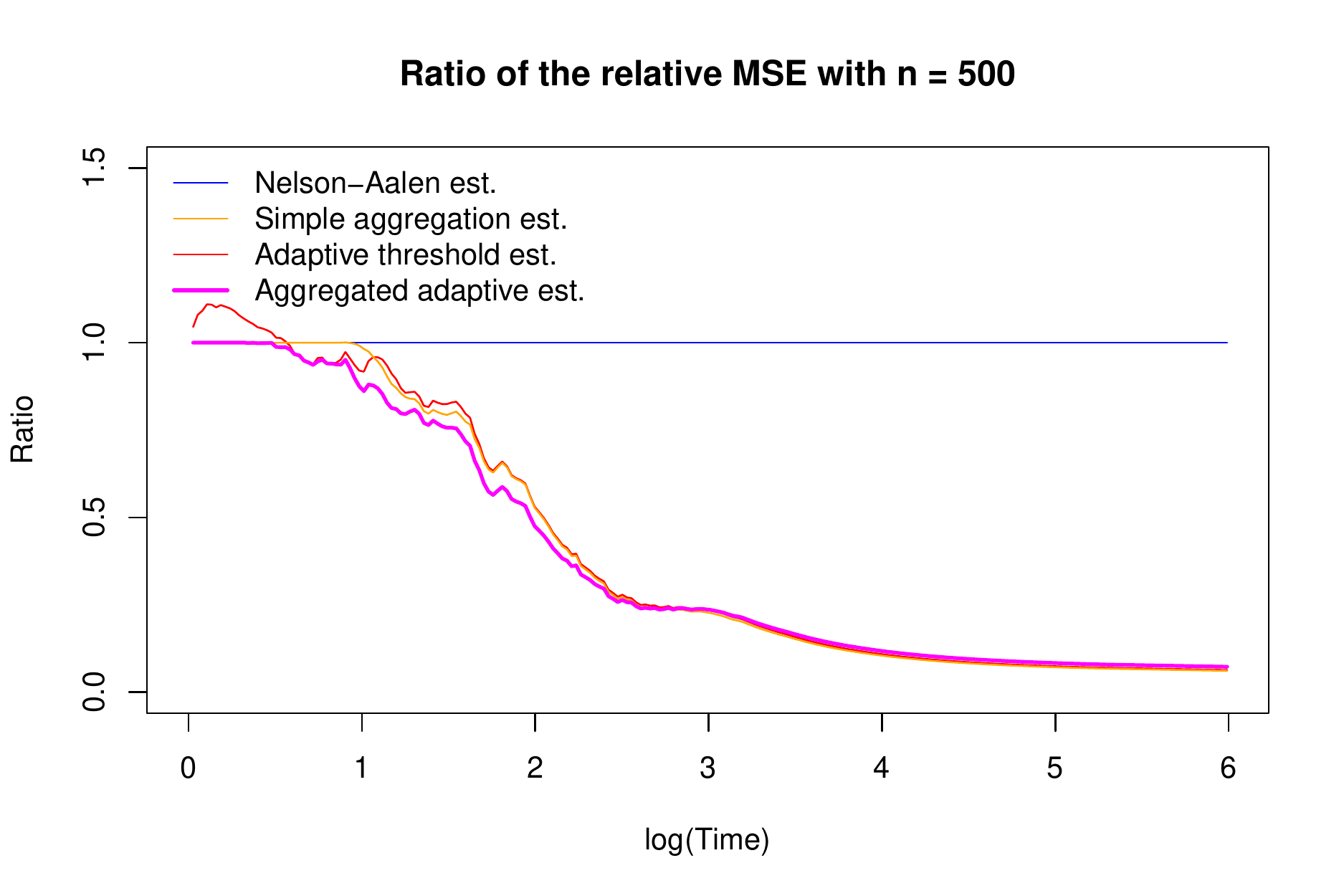}
	\end{tabular}}
	\caption{Relative MSE (left) and ratio of the relative MSE (right) for $n=100$ (top) and $n=500$ (bottom) and $NMC = 1000$. The censorship follows a log-Gamma distribution with $a=5$ and $b=3.5$.}
	\label{SimLogGammaFig}
\end{figure}
Figure \ref{SimLogGammaFig} shows the $RelMSE$ and the ratio of $RelMSE$ of the baseline function compared to the Nelson-Aalen estimator where the baseline function follows a log-Gamma distribution with parameters $a=2$ and $b=2$.
In Table \ref{SimLogGamma}, we compared the $RelMSE$ of estimated survival probabilites for extremes values.
\begin{table}[h]
	\begin{center} 
		\begin{tabular}{|c|c|c|c|c|c|}
			\hline
			\multicolumn{6}{|c|}{$n=100$}\\
			\hline
			\hline
			$x$ & $100$ & $200$ & $300$ & $400$ & $500$\\
			\hline
			%$S_0(x)$ & $0.5570$ & $0.0503$&$0.0208$&$0.0075$\\
			%Mean of $\widehat{S}_0(x)$ & $0.5586$ & $0.0897$&$0.0870$&$0.0865$\\
			%Mean of $\widehat{S}_0(\given{x}{\tau,\beta,z})$ & $0.5180$ & $0.0302$&$0.0009$&$0.00003$\\
			%Mean of $\widehat{S}_0(\given{x}{\widehat \tau ,\beta,z})$ & $0.5618$ & $0.0309$&$0.0067$&$0.0013$\\
			RelMSE of $\widehat{S}_0(x)$ &3.7469&  7.6762&15.7922& 27.1767& 41.9159\\
			%MARE of $\widehat{S}_0(\given{x}{\tau,\beta,z})$ & $0.0799$ & $0.4324$&$0.9356$&$0.9937$\\
			RelMSE of $\widehat{S}_{0,\widehat \tau ,\beta}(x)$ & 2.5996&  4.0775&  6.7142&10.1159& 14.3147\\
			RelMSE with simple aggregation&0.7026& 1.0984& 1.8291& 2.8042& 4.0424\\
			RelMSE with adaptive aggregation&0.7087 &1.0536 &1.6645& 2.4551& 3.4391 \\
			\hline
			\hline
			\multicolumn{6}{|c|}{$n=500$}\\
			\hline
			\hline
			$x$ & $100$ & $200$ & $300$ & $400$&$500$\\
			\hline
			RelMSE of $\widehat{S}_0(x)$ & 1.7343&  4.3083& 10.5811& 20.0979& 32.9481\\
			%MARE of $\widehat{S}_0(\given{x}{\tau,\beta,z})$ & $0.0799$ & $0.4324$&$0.9356$&$0.9937$\\
			RelMSE of $\widehat{S}_{0,\widehat \tau ,\beta}(x)$ &0.2985& 0.4961& 0.8774 &1.4057 &2.0957\\
			RelMSE with simple aggregation& 0.2876& 0.4771& 0.8439& 1.3538& 2.0217\\
			RelMSE with adaptive aggregation&0.3110 &0.5347 &0.9744& 1.5906 &2.4005\\
			\hline
		\end{tabular}
	\end{center}
	\caption{$1000$ Monte-Carlo simulations with the censorship distribution following a log-Gamma distribution with $a=5$ and $b=3.5$.}
	\label{SimLogGamma}
\end{table}

\section{Applications}\label{app}

\subsection{Bladder data set}\label{appBlad}

As a second example, we consider the data set \verb+bladder+ included in the R package \verb+survival+ (\url{https://cran.r-project.org/web/packages/survival/index.html}). This data set concerns the comparison of different treatments on the recurrence  of Stage I bladder tumor (see \cite{byar1980} for more details). We study here only the difference between the placebo and the thiotepa treatment. The initial purpose behind the study of this data set was to determine if the treatment had an effect on the recurrence of the bladder tumor. This study has been done in \cite{wei1989} using the usual Cox model. We want to extend the problem by determining the probability of having the first recurrence of the bladder tumor (the first recurrence is the most important to examine the treatment effect) at the end of the study for the placebo and treatment groups or at what time does the estimated probability of having the first recurrence fall below $0.3$.

We consider the observed time as the time between two recurrences or between the last recurrence and the censoring time. The covariate includes the treatment, the number of initial tumors, the size of the initial tumor and the number of recurrences.

\begin{figure}[h]
	\makebox[\textwidth][c]{\begin{tabular}{cc}
			\includegraphics[width=70mm,height=52.5mm]{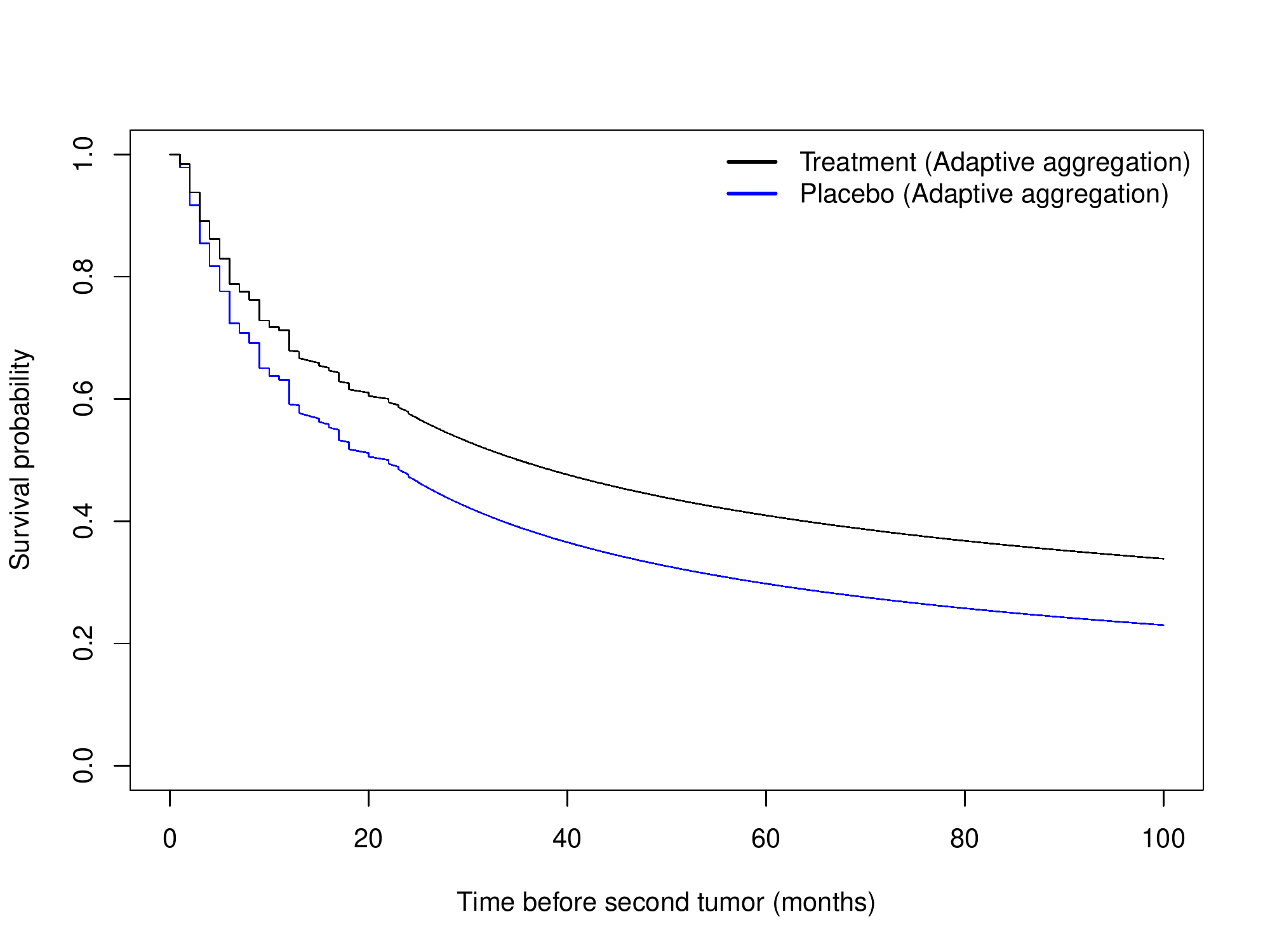}&\includegraphics[width=70mm,height=52.5mm]{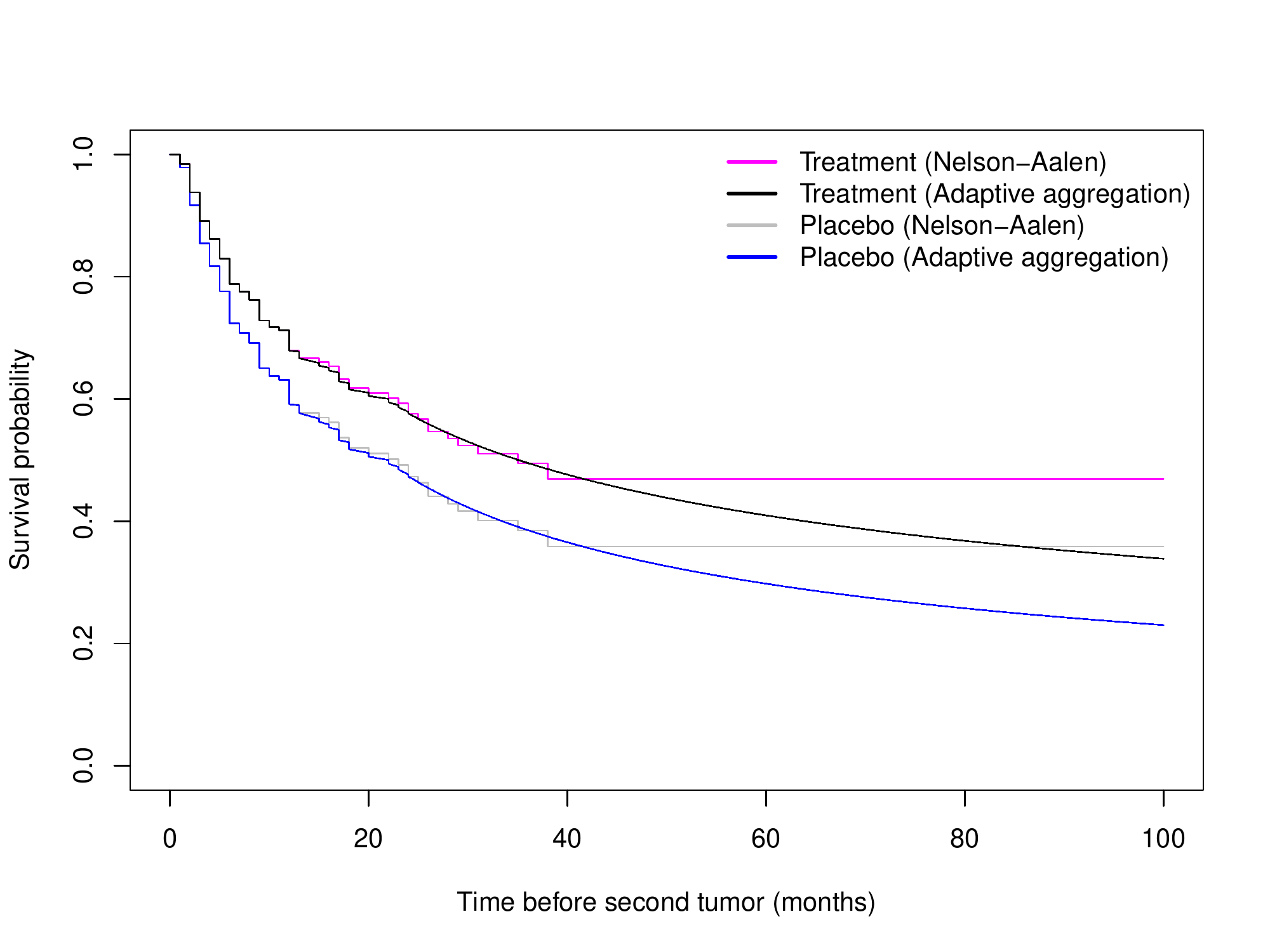}
	\end{tabular}}
	\caption{Estimated survival probabilities of the time of having the first recurrence of the bladder tumor when the initial size of the tumor is 1 and the initial number of tumors is 1.}
	\label{bladder1}
\end{figure}

One can see on the Figure \ref{bladder1} that our model fits the tail of the distribution and the recurrence time is estimated after the last observed time. The estimated survival probability of having the first recurrence of the tumor beyond 3 and 4 years are given in the following table.

\begin{table}[H]
	\begin{center}
		\begin{tabular}{|c|c|c|}
			\hline
			\textbf{Time (months)} &$72$ &$96$\\
			\hline
			\textbf{Nelson-Aalen estimator (Placebo)} &$0.3586$&$0.3586$\\
			\hline
			\textbf{Nelson-Aalen estimator (Treatment)}  & $0.4697$&$0.4697$\\
			\hline
			\textbf{Adaptive aggregation estimator (Placebo)} &$0.2715$&$0.2348$\\
			\hline
			\textbf{Adaptive aggregation estimator (Treatment)}  & $0.3826$&$0.3438$\\
			\hline
		\end{tabular}
	\end{center}
	\caption{Estimated probabilities of having the first recurrence of the tumor beyond 3 and 4 years when the initial size of the tumor is 1, the initial number of tumors is 1.}
	\label{Bladd1}
\end{table}

The next table gives the estimated time at which a patient has a survival probability of having a first recurrence of $0.3$ and $0.4$.

\begin{table}[H]
	\begin{center}
		\begin{tabular}{|c|c|c|}
			\hline
			\textbf{Survival probability} &$0.3$ &$0.4$\\
			\hline
			\textbf{Nelson-Aalen estimator (Placebo)} &$NA$&$35$\\
			\hline
			\textbf{Nelson-Aalen estimator (Treatment)}  & $NA$&$NA$\\
			\hline
			\textbf{Adaptive aggregation estimator (Placebo)} &$95.6232$&$43.3311$\\
			\hline
			\textbf{Adaptive aggregation estimator (Treatment)}  & $194.9080$&$66.58699$\\
			\hline
		\end{tabular}
	\end{center}
	\caption{Estimated time (months) of having the first recurrence of the tumor with a survival probabilities of $0.3$ and $0.4$ when the initial size of the tumor is 1, the initial number of tumors is 1.}
	\label{Bladd2}
\end{table}

One can see that with the model proposed in this paper, the initial problem of analyzing the estimated regression parameter to observe an effect of the treatment has been extended. It is possible to give an estimated probability of having the recurrence before a certain time and it's possible to give an estimated time of recurrence for a given probability.

\subsection{Application to electric consumption prediction}\label{appElec}

In order to offer an alternative to load shedding in case of an electric constraint, a research project has been conducted in Lorient, France, to study the electric consumption of households. 
The idea is to allocate the available electricity among all the consumers in the concerned area by lowering their usual power and so, avoiding the black out. This action can be done remotely from the smart meter. One of the objectives of the experiment is to study the behaviour of the consumers and the consequences on their consumption.
Of course, the process will be used only when the market offers can not cope anymore.  
This action is known as 'active power modulation'.
The data are collected on selected houses to study the effect during the electric constraint. For example, if a house with a maximal electric power contract of $9$ kiloVolt Ampere has a constraint of $50\%$, the maximal electric power becomes $4.5$ kVA. The goal of this study is to minimise the number of houses without electricity during a major power outages. If the electric power requested by the house exceeds the maximal permitted power, the breaker cuts off and the house has no electricity. In this section, we predict the electric power level for one random house during the time of the constraint and compare the measured level with what really happened.

The data used in this application are the electric power of a house with a maximal power contract of $9$ kVA. A measurement of the electric power is made every 10 minutes and corresponds to the mean load power requested in 10 minutes. The outside temperature is collected at the same times. The study period started on the $23$rd December, 2015 and finished on the $21$st March, 2016. Figure \ref{ExTempElec} shows the consumption of the studied house during the period and the measured outside temperature.

\begin{figure}[h]
	\begin{center} 
		\begin{tabular}{c}
			\includegraphics[width=120mm,height=60mm]{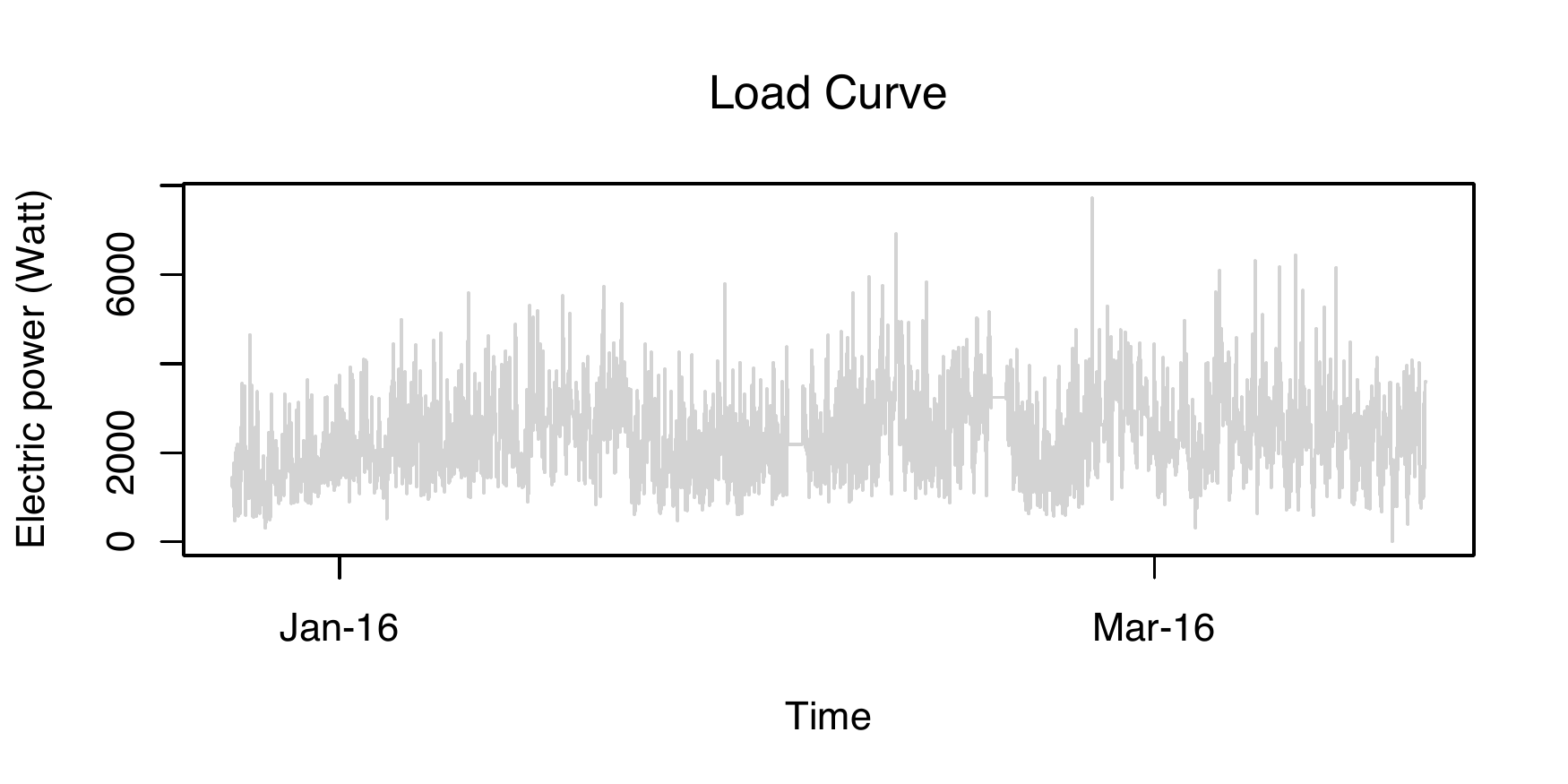} \\ \includegraphics[width=120mm,height=60mm]{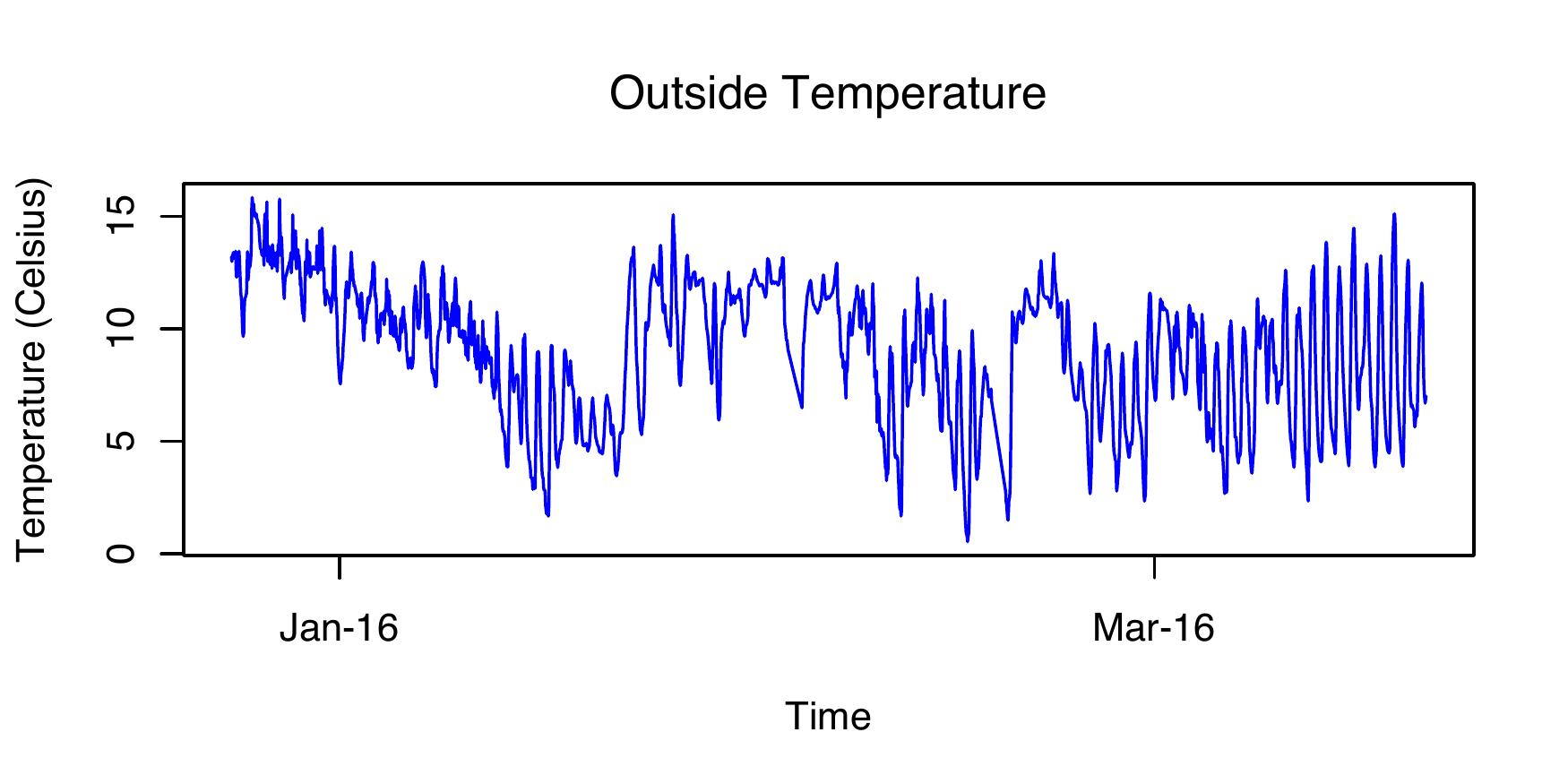}
		\end{tabular}
		\caption{Load curve (top) and outside temperature (bottom) for the studied period.}
		\label{ExTempElec}
	\end{center}
\end{figure}

As one can see in the Figure \ref{ExTempElec}, we deal with a time series. We decided to remove the dependence of the time by discretized the data by hour, under the hypothesis that during the winter, the distribution of the electric power during the same hours remains similar over the days. The hours become part of the covariate and a binary information is given, e.g. a measurement between $2$h et $3$h will have a $1$ in during this hour and $0$ elsewhere.

Moreover, the temperature is included in the covariate with a subtle transformation. Indeed, we separated the temperature into $4$ linear covariates as we assume that the parameter of the temperature will not be constant over the scope of the temperature.

%First, we must determine if the proportional hazard assumption is verified. An easy way to see if the assumption holds is to draw the empirical hazard functions for each class of hours.

For this data set, the cumulative hazard rate functions are not proportionals. We decided to separate the data into five classes to improve the estimation of extreme probabilities. For each group corresponds a period during the day. The hour classes are from $22$h until $6$h which corresponds to the night. From $6$h until $10$h, which corresponds to the morning. From $10$h until $14$h, which corresponds to the lunch time. From $14$h until $18$h, which corresponds to the afternoon and finally from $18$h until $22$h, which corresponds to the evening.

%\begin{table}[H]
%	\caption{Cumulative hazard rate functions for the different groups of hours.}
%	\label{HazGroups}
%	\hspace{-1cm}
%	\begin{tabular}{cc}
%		\includegraphics[width=80mm,height=60mm]{H610.png} & \includegraphics[width=80mm,height=60mm]{H1014.png}\\
%		\includegraphics[width=80mm,height=60mm]{H1418.png} & \includegraphics[width=80mm,height=60mm]{H1822.png}\\
%		\includegraphics[width=80mm,height=60mm]{H226.png} &
%	\end{tabular}
%\end{table}

The proportional hazards assumption almost holds for each groups. We are interested in the impact of the temperature onto this assumption, but the size of the data is not big enough to verify this.

Using the hypothesis from which the distribution of the electric power during the same hours is similar over the days during the winter, we estimate the survival functions. The goal is to predict the probability to exceed the maximal authorized power during the time of the constraint. Table \ref{ecrtable} shows the different time of the constraint, the value of  the maximal power during the constraint and the average outside temperature during the period starting the $23$rd December, 2015 and finishing the $21$st March, 2016. Recall that the maximal power of this house when there is no constraint is $9$ kVA. For each constraint, we give the estimated survival probability to exceed the maximal power given the time and the outside temperature.

\begin{table}[h]
	\begin{center}
		\makebox[\textwidth][c]{\begin{tabular}{|c|c|c|c|c|c|c|c|}
				\hline
				\textbf{Constraint day} &$11$th January &$13$th January&$18$th January&$25$th February\\
				\textbf{Constraint hours} &$17$-$19$h &$14$-$18$h&$14$-$18$h&$14$-$18$h\\
				\hline
				\textbf{Maximal power}  &  $6.3$ kVA & $4.5$ kVA& $4.5$ kVA& $4.5$ kVA\\
				\hline
				\textbf{Outside temperature} &$8.56$ T°C&$8.85$ T°C&$6.82$ T°C&$9.85$ T°C \\
				\hline
				\textbf{Estimated survival}  &  \multirow{ 2}{*}{$0.0015$} & \multirow{ 2}{*}{$0.09$} & \multirow{ 2}{*}{$0.1188$} & \multirow{ 2}{*}{$0.0778$} \\
				\textbf{probability}&&&&\\
				\hline
				\textbf{Number of cut off}  &  \multirow{ 2}{*}{$0$} & \multirow{ 2}{*}{$0$} & \multirow{ 2}{*}{$0$} & \multirow{ 2}{*}{$0$} \\
				\textbf{of the breaker}&&&&\\
				\hline
				\hline
				\textbf{Constraint day} &$1$st March&$7$th March& $18$th March&\\
				\textbf{Constraint hours} &$17$-$19$h&$14$-$18$h& $10$-$12$h&\\
				\hline
				\textbf{Maximal power}  & $3.6$ kVA& $2.8$ kVA& $2.8$ kVA&\\
				\hline
				\textbf{Outside temperature} &$10.92$ T°C&$9.70$ T°C&$9.98$ T°C& \\
				\hline
				\textbf{Estimated survival}  & \multirow{ 2}{*}{$0.0526$} & \multirow{ 2}{*}{$0.5018$} & \multirow{ 2}{*}{$0.7101$}& \\
				\textbf{probability}&&&&\\
				\hline
				\textbf{Number of cut off}  &  \multirow{ 2}{*}{$1$} & \multirow{ 2}{*}{$2$} & \multirow{ 2}{*}{$8$} &  \\
				\textbf{of the breaker}&&&&\\
				\hline
		\end{tabular}}
	\end{center}
	\caption{Date of the electric constraint, maximal power, outside temperature and estimated survival probabilities during the electric constraint. The number of cut off of the breaker represents the number of time the breaker cut off during the constraint period.}
	\label{ecrtable}
\end{table}

The probability corresponding to a return period of $4$ hours (happen once in any given $4$ hours period) is $\frac{1}{24} \simeq 0.04$. For a return period of $2$ hours, the corresponding probability is $0.08$. We can see on Table \ref{ecrtable} that we detect five estimated probabilities exceeding the probability associated to the return period. These probabilities correspond to the constraints of the $13$th of January, the $18$th of January, the $15$th of February, the $7$th of March and the $18$th of March. This house is therefore considered at risk during these five constraints. 

We are now interested to see what really happened for this house during these constraints. This house had one opening of its breaker during the constraint of the $1$st of March, two during the constraint of the $7$th of March, eight during the constraint of the $18$th of March and none during the other constraints. We can see that we have very high estimated probabilities of having one observation exceeding the maximal power during constraint for the $7$th and $18$th of March and that the house had multiple cuts off during these constraints. We can suppose that during the other constraints, the house anticipated the constraint period and reduced its electric power by changing its behavior.

\section{Conclusion}

In this article, we propose an extension of the Cox model in order to estimate probabilities of rare events and extreme quantiles. The model is semi-parametric and composed of the Nelson-Aalen estimator for the non-parametric part and the parametric part is described by a Pareto distribution. We prove the consistency of the estimator of the Pareto parameter and give an explicit convergence rate for the Hall model.

A data-driven choice of the threshold motivated by a goodness-of-fit test is proposed. 
An aggregated estimator and an adaptive aggregated estimator are suggested and studied in order
% extension of the model is proposed as an aggregation of the estimated cumulative hazard functions 
to improve the fitness of the model onto the data.
The performance of the proposed estimators is demonstrated on artificial data.

Two applications on real data sets are given. The application on the bladder data shows the motive of the model as it allows an estimation of extreme quantiles which was not possible with the usual Cox model. The application on the electric consumption gives an application onto data where the main purpose is to estimate survival probabilities and we are not interested to test if there is an effect of a treatment.

\appendix

\section{Proofs of the results.}

\subsection{Proof of Theorem \ref{thp1}.}
In the sequel we denote quasi-log-likelihood ratio by 
\begin{align*} %\label{}
\mathcal{L}(\given{\theta',\theta}{\mathbf z})
	:=\mathcal{L}(\given{\theta'}{\mathbf z}) - \mathcal{L}(\given{\theta}{\mathbf z}) %\label{likdiffer001}
	= \sum_{i=1}^n \ln \frac{dP_{S_{0,\tau,\theta'}}}{dP_{S_{0,\tau,\theta}}} \left(\given{t_i,\delta_i}{z_i} \right).
\end{align*}
Recall that we denoted by $\mathbb P$ the probability measure corresponding to the "true" model which has 
the baseline survival function $S_0$. 
\begin{lemma}
	\label{L1}
	For any $\theta, \theta' \in \mathbb{R}$, $\mathbf{z} \in \mathbb{Z}$ and any $x>0$, it holds
	\begin{equation*}
	\mathbb{P}(
	\mathcal{L}(\given{\theta',\theta}{\mathbf z})
	%\mathcal{L}(\given{\theta'}{\mathbf{z}}) - \mathcal{L}(\given{\theta}{\mathbf{z}}) 
	>x+\sum_{i=1}^n\chi^2(P_{S_0}(\given{.}{z_i}),P_{S_{0,\tau,\theta}}(\given{.}{z_i}))) \leq e^{-x/2}
	\end{equation*}
\end{lemma}
\begin{proof}
	Let $P_{S_{0,\tau,\theta'}}(\given{.}{z_i})$ and $P_{S_{0,\tau,\theta}}(\given{.}{z_i})$ be the conditional cumulative distribution function of $( T, \Delta )$ given $Z=z_i$ where the survival function has a Pareto tail with parameter $\theta'$ and $\theta$ respectively for a threshold $\tau$. 
%	The quasi-log-likelihood ratio is given by
%	\begin{equation*}
%	\mathcal{L}(\given{\theta',\theta}{\mathbf z})
%	%\mathcal{L}\left(P_{S_{0,\tau,\theta'}},P_{S_{0,\tau,\theta}} \right) 
%	= \sum_{i=1}^n \ln \frac{dP_{S_{0,\tau,\theta'}}}{dP_{S_{0,\tau,\theta}}} \left(\given{t_i,\delta_i}{z_i} \right).
%	\end{equation*}
	By Chebychev's exponential inequality, we have for any $y>0$:
	\begin{align*}
	\mathbb{P}\left(
	\mathcal{L}(\given{\theta',\theta}{\mathbf z}) %\mathcal{L}\left(P_{S_{0,\tau,\theta'}},P_{S_{0,\tau,\theta}} \right) 
	>y \right) 
	&\leq e^{-y/2} \mathbb{E}\left( e^{\frac{1}{2} 
	\mathcal{L}(\given{\theta',\theta}{\mathbf z})
	%\mathcal{L}\left(P_{S_{0,\tau,\theta'}},P_{S_{0,\tau,\theta}} \right) 
	   }     \right)  \\
	%	\end{equation*}
	%	Using Lagrange's transform yields 
	%	\begin{equation*}
	%	\mathbb{P}\left(\mathcal{L}\left(P_{S_{0,\tau,\theta}},P_{S_{0,\tau,\theta'}}\right) >y \right) 
	&\leq \exp\left[-y/2 + \ln \left( \mathbb{E}\left( e^{\frac{1}{2}
	\mathcal{L}(\given{\theta',\theta}{\mathbf z})}
	%\mathcal{L}\left(P_{S_{0,\tau,\theta'}},P_{S_{0,\tau,\theta}} \right)
	\right) \right)\right].
	\end{align*}
	As the triplet $\{\given{t_1,\delta_1}{z_1} \},...,\{\given{t_n,\delta_n}{z_n}\}$ are independent, we can write the term 
$\ln \left(\mathbb{E}\left( e^{\frac{1}{2}
\mathcal{L}(\given{\theta',\theta}{\mathbf z})}%\mathcal{L}\left(P_{S_{0,\tau,\theta'}},P_{S_{0,\tau,\theta}} \right)
 \right)\right)$ 
as
	\begin{align*}
	\ln \left(\mathbb{E}\left( e^{\frac{1}{2}
	\mathcal{L}(\given{\theta',\theta}{\mathbf z})}%\mathcal{L}\left(P_{S_{0,\tau,\theta'}},P_{S_{0,\tau,\theta}} \right)
	 \right)\right) 
	&= \ln \left(\mathbb{E}\left( e^{\frac{1}{2}\sum_{i=1}^n \ln \frac{dP_{S_{0,\tau,\theta'}}}{dP_{S_{0,\tau,\theta}}} \left(\given{t_i,\delta_i}{z_i} \right)} \right)\right)\\
	&= \ln \left(\mathbb{E}\left( \prod_{i=1}^n  \sqrt{\frac{dP_{S_{0,\tau,\theta'}}}{dP_{S_{0,\tau,\theta}}} \left(\given{t_i,\delta_i}{z_i} \right)} \right)\right)\\
	&= \sum_{i=1}^n  \ln \left(\mathbb{E}\left( \sqrt{\frac{dP_{S_{0,\tau,\theta'}}}{dP_{S_{0,\tau,\theta}}} \left(\given{t_i,\delta_i}{z_i} \right)} \right)\right).
	\end{align*}
	By H\"older's inequality, we have
	\begin{align*}
	\mathbb{E}\left( \sqrt{\frac{dP_{S_{0,\tau,\theta'}}}{dP_{S_{0,\tau,\theta}}} \left(\given{t_i,\delta_i}{z_i} \right)} \right) &\leq \sqrt{\mathbb{E}\left( \frac{dP_{S_{0,\tau,\theta'}}}{dP_{S_{0,\tau,\theta}}} \left(\given{t_i,\delta_i}{z_i} \right)\right)}\\ 
	&\leq \sqrt{1+\chi^2 (P_{S_0} \left(\given{\cdot}{z_i} \right),P_{S_{0,\tau,\theta}} \left(\given{\cdot}{z_i} \right) )},
	\end{align*}
	where the $\chi^2$ entropy between two equivalent probability measure is defined by \eqref{Xi2 entropy}. %$\chi^2 (P,P_0) = \int dP/dP_0dP - 1$. 
	Then,
	\begin{align*}
	\mathbb{P}\left(
	\mathcal{L}(\given{\theta',\theta}{\mathbf z})%\mathcal{L}\left(P_{S_{0,\tau,\theta'}},P_{S_{0,\tau,\theta}} \right) 
	>y \right) &\leq \exp\left[-\frac{y}{2} + \frac{1}{2}\sum_{i=1}^n \ln\left(1+\chi^2 (P_{S_0} \left(\given{\cdot}{z_i} \right),P_{S_{0,\tau,\theta}} \left(\given{\cdot}{z_i} \right) )\right)\right]\\
	&\leq \exp\left[-\frac{y}{2} + \frac{1}{2}\sum_{i=1}^n \chi^2 (P_{S_0} \left(\given{\cdot}{z_i} \right),P_{S_{0,\tau,\theta}} \left(\given{\cdot}{z_i} \right) )\right].
	\end{align*}
	Setting $y = x + \sum_{i=1}^n \chi^2 (P_{S_0} \left(\given{\cdot}{z_i} \right),P_{S_{0,\tau,\theta}} \left(\given{\cdot}{z_i} \right) )$, gives Lemma \ref{L1}.
\end{proof}

Recall that the Kullback-Leibler divergence $\mathcal{K}(\theta',\theta)$ 
between two Pareto distribution with parameters $\theta'$ and $\theta$ is defined by \eqref{KL entropy}.
%denoted by
%$\mathcal{K}(\theta',\theta)=\int \ln(dG_{\theta'}/dG_\theta)dG_{\theta'}.$  
%It can be written as 
%$$
%\mathcal{K}(\theta',\theta) = \frac{\theta'}{\theta} - 1 - \ln(\frac{\theta'}{\theta}) \sim \left(  \frac{\theta'}{\theta}-1\right)^2
%\quad\text{as} \quad \frac{\theta'}{\theta}\to1.
%$$ 
The following lemma gives %a bound of the maximum likelihood ratio which 
the rate of convergence for $\widehat \theta_\tau$. 
We adapt the proof from \cite{grama2014} to the case of the Cox model under consideration in this paper. 
%$\widehat{n}_\tau$.
\begin{lemma}
	\label{L2}
	For any $\theta >0$, $\tau > x_0$ and $v>0$, it holds
	$$
	\mathbb{P}\left(\widehat{n}_\tau\mathcal{K}(\widehat \theta_\tau,\theta) > v+\sum_{i=1}^n \chi^2 (P_{S_0} \left(\given{\cdot}{z_i} \right),P_{S_{0,\tau,\theta}} \left(\given{\cdot}{z_i} \right) )+2\ln(n)\right) \leq 2\exp(-v/2),
	$$
	where $\widehat{n}_\tau = \sum_{t_i>\tau}\delta_i$.
\end{lemma}
\begin{proof}
	It is easy to see that the assertion of lemma follows from the following bound:
	\begin{equation}
	\label{PL2}
	\mathbb{P}\left( \widehat n_\tau \mathcal{K}(\widehat \theta_\tau,\theta) >v + \sum_{i=1}^n \chi^2 (P_{S_0} \left(\given{\cdot}{z_i} \right),P_{S_{0,\tau,\theta}} \left(\given{\cdot}{z_i} \right) ) \right) \leq 2n\exp\left( -\frac{v}{2} \right).
	\end{equation}
	The likelihood ratio is given by
	\begin{align}
	\mathcal{L}(\given{\theta',\theta}{\mathbf z})
	&:=\mathcal{L}(\given{\theta'}{\mathbf z}) - \mathcal{L}(\given{\theta}{\mathbf z}) \label{likdiffer001} \\
	&= \sum_{i=1}^n \ln p_{S_{0,\tau,\theta'}}(\given{t_i,\delta_i}{z_i}) - \sum_{i=1}^n \ln p_{S_{0,\tau,\theta}}(\given{t_i,\delta_i}{z_i}). \nonumber
	\end{align}
	Then, removing the censoring part from the likelihood ratio, we have 
	\begin{align*}
	\mathcal{L}(\given{\theta',\theta}{\mathbf z})
	%\mathcal{L}^{part}(\given{\theta'}{z_1,\dots,z_n}) - \mathcal{L}^{part}(\given{\theta}{z_1,\dots,z_n} ) \\
	& =  \sum_{i=1}^n \ln(h_{z_i,\tau,\theta'}(t_i)^{\delta_i}S_{{z_i,\tau,\theta'}}(t_i)) - \sum_{i=1}^n \ln(h_{{z_i,\tau,\theta}}(t_i)^{\delta_i}S_{{z_,\tau,\theta}}(t_i))\\
	& =  \sum_{i=1}^n \delta_i\ln(h_{0,\tau,\theta'}(t_i)e^{\beta \cdot z_i})+ \ln(S_{z_i,\tau,\theta'}(t_i))\\
	& - \delta_i\ln(h_{0,\tau,\theta}(t_i)e^{\beta \cdot z_i}) - \ln(S_{z_i,\tau,\theta}(\given{t_i}{z_i}))).
	\end{align*}
	Developing the terms, we obtain	
	\begin{align*}
	\mathcal{L}(\given{\theta',\theta}{\mathbf z})
	&=  \sum_{i=1}^n \delta_i \ln(h_0(t_i))\mathbb{1}_{t_i \leq \tau} + \delta_i \ln(\frac{1}{\theta' t_i})\mathbb{1}_{t_i > \tau} 
	+ \delta_i\beta \cdot z_i \\
	& \qquad\qquad\qquad\qquad\qquad\qquad\qquad\qquad  + e^{\beta \cdot z_i}\ln(S_{0,\tau,\theta'}(t_i))\\
	&   - \delta_i \ln(h_0(t_i))\mathbb{1}_{t_i \leq \tau} -  \delta_i \ln(\frac{1}{\theta t_i})\mathbb{1}_{t_i > \tau} - \delta_i\beta \cdot z_i - e^{\beta \cdot z_i}\ln(S_{0,\tau,\theta}(t_i)),
	%	\\
	%	& =  \sum_{i=1}^n \delta_i ( \ln(\frac{1}{\theta'}) + \ln(\frac{1}{t_i}))\mathbb{1}_{t_i > \tau} - \delta_i( \ln(\frac{1}{\theta}) + \ln(\frac{1}{t_i}))\mathbb{1}_{t_i > \tau} + %e^{\beta \cdot z_i} \ln(S_0(t_i))\mathbb{1}_{t_i \leq \tau}\\
	%	& - e^{\beta \cdot z_i} \ln(S_0(t_i))\mathbb{1}_{t_i \leq \tau} + e^{\beta \cdot z_i} \ln(S_0(\tau)(\frac{t_i}{\tau})^{-1/\theta'})\mathbb{1}_{t_i > \tau} - e^{\beta \cdot z_i} %\ln(S_0(\tau)(\frac{t_i}{\tau})^{-1/\theta})\mathbb{1}_{t_i > \tau}
	\end{align*}
	and further
	\begin{align*}
	\mathcal{L}(\given{\theta',\theta}{\mathbf z})
	& =  \sum_{i=1}^n \left[  \delta_i \ln \left( \frac{\theta}{\theta'} \right) - e^{\beta \cdot z_i} \frac{1}{\theta'}\ln \left( \frac{t_i}{\tau}\right) + e^{\beta \cdot z_i} \frac{1}{\theta}\ln \left( \frac{t_i}{\tau}\right) \right]\mathbb{1}_{t_i > \tau} \\
	& =  \sum_{i=1}^n \left[  \delta_i \ln \left( \frac{\theta}{\theta'} \right) + e^{\beta \cdot z_i} \ln \left( \frac{t_i}{\tau}\right) \left( \frac{1}{\theta} - \frac{1}{\theta'} \right) \right]\mathbb{1}_{t_i > \tau}.
	\end{align*}
	Since $$ \widehat \theta_\tau = \frac{\sum_{i:t_i>\tau} e^{\beta \cdot z_i \ln\left( \frac{t_i}{\tau} \right)}}{\sum_{i:t_i>\tau} \delta_i},$$ 
	using \eqref{likdiffer001}, it follows that
	\begin{align}
	\mathcal{L}(\given{\theta',\theta}{\mathbf z})
	&:=\mathcal{L}(\given{\theta'}{\mathbf z}) - \mathcal{L}(\given{\theta}{\mathbf z}) \nonumber  \\% \label{likdiffer001} \\
	%	&\mathcal{L}^{part}(\given{\theta'}{z_1,\dots,z_n}) - \mathcal{L}^{part}(\given{\theta}{z_1,\dots,z_n}) \\
%	&= \mathcal{L}^{part}(\given{\theta',\theta}{\mathbf z}) \nonumber \\
	& =  \sum_{i=1}^n \mathbb{1}_{t_i > \tau} \delta_i \ln \left( \frac{\theta}{\theta'} \right) + \sum_{i=1}^n \mathbb{1}_{t_i > \tau} \delta_i \widehat \theta_\tau \left( \frac{1}{\theta} - \frac{1}{\theta'} \right) \nonumber  \\
	& =  \sum_{i=1}^n \mathbb{1}_{t_i > \tau} \delta_i \left[ \ln \left( \frac{\theta}{\theta'} \right) + \widehat \theta_\tau \left( \frac{1}{\theta} - \frac{1}{\theta'} \right) \right] \nonumber \\
	& =  \widehat n_\tau \left( \ln \left( \frac{\theta}{\theta'} \right) +  \left( \frac{1}{\theta} - \frac{1}{\theta'} \right)\widehat \theta_\tau \right) \nonumber \\
	&= \widehat n_\tau \Lambda (\theta'),\label{LLL001}
	\end{align}
	where we  denoted for brevity $\Lambda (\theta') = \ln \left( \frac{\theta}{\theta'} \right) - \left( \frac{1}{\theta'} - \frac{1}{\theta} \right)\widehat \theta_\tau$.
	Since $\mathcal{K}(\theta',\theta)=\frac{\theta'}{\theta} - 1 - \ln(\frac{\theta'}{\theta})$, 
	we then have the identity 
\begin{align} \label{Kullback001}
	\mathcal{K}(\widehat \theta_\tau,\theta) = \Lambda (\widehat \theta_\tau).
\end{align}
Using \eqref{LLL001} and \eqref{Kullback001} we hope to bound $\mathcal{K}(\widehat \theta_\tau,\theta)$ in probability by Lemma \ref{L1}.
The problem is that $\widehat \theta_\tau$ is random and therefore we cannot apply Lemma \ref{L1} directly. 
% We would like to use Lemma \ref{L1} to bound  
We shall circumvent this difficulty in the following way.
	For any $y>0$ and any $k\geq 1$, the inequality  $k\mathcal{K}(u,\theta) > y$ can be equivalently written as 
	\begin{align*}
	%	\ln \left( \frac{\theta}{u} \right) -  \left( \frac{1}{u} - \frac{1}{\theta}  \right)\widehat \theta_\tau & >  \frac{y}{k} \\
	\left( \frac{1}{u} - \frac{1}{\theta}  \right)\widehat \theta_\tau & <  - \frac{y}{k} + \ln \left( \frac{\theta}{u} \right).
	\end{align*}
	Setting $g(u,k) = \frac{\ln \left( \frac{\theta}{u} \right) - \frac{y}{k}}{\left( \frac{1}{u} - \frac{1}{\theta}  \right)}$, we have $g(u,k)>\widehat \theta_\tau$, when $0<u<\theta$ and $g(u,k)<\widehat \theta_\tau$, when $u>\theta$. 
	Moreover, the function $g(u,k)$ has a maximum for $0<u<\theta$ and a minimum for $u>\theta$. 
	When $0<u<\theta$, we have $\lim\limits_{u \rightarrow 0^+} \frac{dg(u,k)}{du} = + \infty$ and  $\lim\limits_{u \rightarrow \theta^-} \frac{dg(u,k)}{du} = - \infty$. 
	When $u> \theta$, we have $\lim\limits_{u \rightarrow \theta^+} \frac{dg(u,k)}{du} = + \infty$ and $\lim\limits_{u \rightarrow +\infty} \frac{dg(u,k)}{du} = 0^-$. 
	Let 
	$$\theta^+ (k) = \underset{0<u<\theta}{\operatorname{ {\mbox{argmax} } }}\ g(u,k) \ \ \text{and} \ \  
	\theta^- (k) = \underset{u>\theta}{\operatorname{\mbox{argmin} }}\ g(u,k), 
	$$ 
	then 
	\begin{align*}
	\{ \widehat n_\tau \Lambda(\widehat \theta_\tau)>y, \widehat \theta_\tau < \theta \}& =  \{ g(\widehat \theta_\tau,\widehat n_\tau) >\widehat \theta_\tau, \widehat \theta_\tau < \theta \} \\
	& \subset  \{ g(\theta^+(\widehat n_\tau),\widehat n_\tau) >\widehat \theta_\tau, \widehat \theta_\tau < \theta \}\\
	& =  \{ \widehat n_\tau \Lambda(\theta^+(\widehat n_\tau))>y, \widehat \theta_\tau < \theta \}\\
	& \subset  \{ \widehat n_\tau \Lambda(\theta^+(\widehat n_\tau))>y \}.
	\end{align*}
	In the same way, we have $\{ \widehat n_\tau \Lambda(\widehat \theta_\tau)>y, \widehat \theta_\tau > \theta \} \subset \{ \widehat n_\tau \Lambda(\theta^-(\widehat n_\tau))>y \}$. Since $\mathcal{K}(\widehat \theta_\tau,\theta) = \Lambda (\widehat \theta_\tau)$, this implies 
	$$\{ \widehat n_\tau \Lambda(\widehat \theta_\tau)>y \} \subset \{ \widehat n_\tau \Lambda(\theta^+(\widehat n_\tau))>y \} \cup \{ \widehat n_\tau \Lambda(\theta^-(\widehat n_\tau))>y \}.$$
	We then have
	\begin{align}
	\mathbb{P}(\widehat n_\tau \mathcal{K}(\widehat \theta_\tau,\theta) >y) 
	&\leq \mathbb{P}(\widehat n_\tau \Lambda(\theta^+(\widehat n_\tau))>y)  + \mathbb{P}(\widehat n_\tau \Lambda(\theta^-(\widehat n_\tau))>y) \nonumber \\
	&\leq \sum_{k=1}^n \mathbb{P}(\widehat n_\tau \Lambda(\theta^+(k))>y)  + \sum_{k=1}^n \mathbb{P}(\widehat n_\tau \Lambda(\theta^-(k))>y). \label{FFF001}
	\end{align}
	Now we can apply Lemma \ref{L1}, 
	with 	
	\begin{align} \label{YYYY001}
	y=v + \sum_{i=1}^n \chi^2 (P_{S_0} \left(\given{\cdot}{z_i} \right),P_{S_{0,\tau,\theta}} \left(\given{\cdot}{z_i} \right) )  
	\end{align}
	from which it follows that, for $k=1,...,n$,
	\begin{align}
	\mathbb{P}(\widehat n_\tau \Lambda(\theta^{\pm}(k))>y) \leq e^{-v/2}. \label{FFF002}
	\end{align}
From \eqref{FFF001} and \eqref{FFF002},
	\begin{align}
	\mathbb{P}(\widehat n_\tau \mathcal{K}(\widehat \theta_\tau,\theta) >y) \leq 
	2n e^{-v/2}.	\end{align}
%Replacing $y$ by $v+\sum_{i=1}^n\chi^2(P_{S_0}(\given{.}{z_i}),P_{S_{0,\tau,\theta}}(\given{.}{z_i}))$, 
This and \eqref{YYYY001} yields \eqref{PL2}, thus concluding the proof of Lemma \ref{L2}.
\end{proof}
%The notation $a_n = O(b_n)$ means that there is a positive constant $c$ such that $\mathbb{P}(a_n>cb_n,b_n<\infty)\rightarrow 0$ as $n \rightarrow \infty$.
Theorem \ref{thp1} follows, if we set $v = 2\ln(n)$ in Lemma \ref{L2}.

%From Theorem \ref{thp1} and condition (\ref{c1}), we deduce the following
%
%\begin{theorem}
%\label{thp2}
%Assume that the distribution $F$ belongs to the domain of attraction of the Fr?chet distribution and $\tau$ satisfies (\ref{c1}). Then
%$$ 
%\mathcal{K}(\widehat \theta_\tau,\theta) = O_{\mathbb{P}}\left( \frac{\ln n}{\widehat n_\tau} \right)
%$$
%\end{theorem}

\subsection{Verification of Condition \ref{Crate}}

\begin{lemma}
	\label{LCrate}
	Assume that the distribution functions given the covariate $z$
	%\todo{Lemma réécrit (KJ)}
	of both the survival and censoring time are in the maximal domain of attraction of the Fr\'echet law with parameters $\theta(z)$ and $\theta_C(z)$ respectively. Then, for any z,
	$$
	q_{F}(\given{\tau}{z}) \to \frac{\theta(z)} {\theta_C(z) + \theta(z)} \ \text{as} \ \tau \to \infty.
	$$
	where $\theta(z)>0$ and $\theta_C(z) >0$.
\end{lemma}
\begin{proof}
	We have for any z,
	\begin{equation*}
	q_{F}(\given{\tau}{z}) = \int_\tau^{\infty} \frac{S(\given{t}{z})}{S(\given{\tau}{z})} \frac{f_{C}(\given{t}{z})}{S_{C}(\given{\tau}{z})}dt \in [0,1],
	\quad \tau \geq x_0.
	\end{equation*}
	Since $F$ and $F_{C}$ are in the maximal domain of attraction of the Fr\'echet law, for some $\theta(z)>0$ and $\theta_C(z)>0$, we have
	$$
	\frac{S(\given{\tau t}{z})}{S(\given{\tau}{z})} \to t^{-1/\theta(z)}, \ \text{as} \ \tau \to \infty,
	$$
	and
	$$
	\frac{S_{C}(\given{\tau t}{z})}{S_{C}(\given{\tau}{z})} \to t^{-1/\theta_C(z)}, \ \text{as} \ \tau \to \infty.
	$$
	Therefore,
	\begin{equation*}
	\lim_{\tau \to \infty} q_{F}(\given{\tau}{z}) = \lim_{\tau \to \infty}  \int_\tau^{\infty} \frac{S(\given{t}{z})}{S(\given{\tau}{z})} \frac{f_{C}(\given{t}{z})}{S_{C}(\given{\tau}{z})}dt.
	\end{equation*}
	By the Lebesgue theorem of dominated convergence,
	\begin{align*}
	\lim_{\tau \to \infty} q_{F}(\given{\tau}{z}) &=  \int_1^{\infty} t^{-1/\theta(z)} \frac{1}{\theta_C(z) \tau} t^{-1/\theta_C(z) - 1}d\tau t \\
	& = \frac{1}{\theta_C(z)} \int_1^{\infty} t^{-\frac{\theta_C(z) + \theta(z)}{\theta(z) \theta_C(z)} -1} dt\\
	& = \frac{1}{\theta_C(z)} \frac{\theta(z) \theta_C(z)} {\theta_C(z) + \theta(z)}\\
	& = \frac{\theta(z)} {\theta_C(z) + \theta(z)}.
	\end{align*}
\end{proof}

\subsection{Proof of Theorem \ref{th2}}

We begin with an auxiliary theorem.
\begin{theorem}
	\label{th1}
	Assume condition \eqref{Crate} and \eqref{CBound}.
	Then, there exists a constant $c>0$ such that,
	$$
	\lim_{n\to \infty}  \mathbb{P}  \left(\mathcal{K}(\widehat \theta_{\tau},\theta)  
	\leq c \frac{\sum_{i=1}^n \chi^2(P_{S_0}(\given{.}{z_i}),P_{S_{0,\tau,\theta}}(\given{.}{z_i})) + 4\ln(n)}{\sum_{i=1}^nS_C(\given{\tau}{z_i})S(\given{\tau}{z_i})}     \right) 
	= 1,
	$$
	where $P_{S_0}(\given{.}{z_i})$ is the cumulative distribution function of \eqref{dens} and $P_{S_{0,\tau,\theta}}(\given{.}{z_i})$ is the cumulative distribution function of the joint density of the model \eqref{sem-mod}.
\end{theorem}
\begin{proof}
	Theorem \ref{th1} follows from Theorem \ref{thp1} and the following Lemma \ref{L4}.
\end{proof}
\begin{lemma}
	\label{L4}
	Assume condition (\ref{Crate}). Then, for every $\tau \geq x_0$, 
	$$
	\mathbb{E}(\widehat n_\tau) \geq \sum_{i=1}^nS(\given{\tau}{z_i})S_C(\given{\tau}{z_i}) (1-q_0),
	$$
	and
	$$
	\mathbb{P}(\widehat n_\tau \leq \mathbb{E}(\widehat n_\tau)/2) \leq e^{-\mathbb{E}(\widehat n_\tau)/8}.
	$$
\end{lemma}
\begin{proof}
	By the density of the model (\ref{dens}), we have 
	$$
	\mathbb{E}(\widehat n_\tau) =  \sum_{i=1}^n \int_\tau^\infty f(\given{x}{z_i})S_C(\given{x}{z_i})dx,
	$$
	where $\widehat{n}_\tau = \sum_{t_i>\tau}\delta_i$. Therefore, 
	\begin{align*}
	\mathbb{E}(\widehat n_\tau) &= \sum_{i=1}^n\left( \left[- S(\given{x}{z_i})S_C(\given{x}{z_i}) \right]_\tau^\infty - \int_\tau^\infty S(\given{x}{z_i})f_C(\given{x}{z_i})dx\right)\\
	&= \sum_{i=1}^n\left( S(\given{\tau}{z_i})S_C(\given{\tau}{z_i}) - S(\given{\tau}{z_i})S_C(\given{\tau}{z_i})\int_\tau^\infty \frac{S(\given{x}{z_i})}{S(\given{\tau}{z_i})}\frac{f_C(\given{x}{z_i})}{S_C(\given{\tau}{z_i})}dx\right)\\
	&= \sum_{i=1}^nS(\given{\tau}{z_i})S_C(\given{\tau}{z_i}) (1 - q_{F}(\given{\tau}{z_i})).
	\end{align*}
	using $q_{F}(\given{\tau}{z_i}) \leq q_0$ for any $z_i$ proves the first part of the Lemma. 

The second inequality follows from the exponential bound for binomial random variables, and is established by using standard techniques going back to Chernoff \cite{chernoff1952}.	
\end{proof}

\begin{lemma}
	\label{L6}
	Assume that $Q$ and $Q_0$ are two equivalent probability measures on a measurable space. Then, 
	$$
	\chi^2(Q,Q_0) \leq \int \left( \ln\frac{dQ}{dQ_0} \right)^2 \exp\left| \frac{dQ}{dQ_0} \right| dQ.
	$$
\end{lemma}
\begin{proof}
	The proof can be found in the article \cite{grama2014}.
\end{proof}
Now we proceed to prove Theorem \ref{th2}. We start by providing the following bound:
$$
\max_{z \in \mathbb{Z}} \chi^2(P_{S_0}(\given{\cdot}{z}),P_{S_{0,\tau,\theta}}(\given{\cdot}{z})) 
\leq O_{\mathbb{P}}\left( \rho_\tau^2 \max_{z \in \mathbb{Z}} S(\given{\tau}{z}) S_C(\given{\tau}{z}) \right).
$$
%The easiest way to do this is to choose $z$ such as the $\chi^2$ distance between the two measures is maximized for $z \in \mathbb{Z}$. 
%%%%%%%%%%%%%%
%\todo{Move here the bound }
%%%%%%%%%%%%%%%
By Lemma \ref{L6}, we have
\begin{align*}
&\max_{z \in \mathbb{Z}} \chi^2(P_{S_0}(\given{.}{z}),P_{S_{0,\tau,\theta}}(\given{.}{z})) \\
&\leq \max_{z \in \mathbb{Z}} \int  \ln^2\frac{dP_{S_0}(\given{t,\delta}{z})}{dP_{S_{0,\tau,\theta}}(\given{t,\delta}{z})} 
\exp\left| \ln \frac{dP_{S_0}(\given{t,\delta}{z})}{dP_{S_{0,\tau,\theta}}(\given{t,\delta}{z})} \right| P_{S_0}(\given{dt,d\delta}{z}).
\end{align*}
For any $x>\tau$, we have
\begin{align*}
\ln\frac{dP_{S_0}(\given{t,\delta}{z})}{dP_{S_{0,\tau,\theta}}(\given{t,\delta}{z})} 
&= \ln \frac{h(\given{t}{z})^\delta S(\given{t}{z})}{h_\theta(\given{t}{z})^\delta S_{\theta,\tau}(\given{t}{z})} \\
&= \ln \frac{h_0(t)^\delta (e^{\beta \cdot z})^\delta  e^{-\int_\tau^t h_0(u)e^{\beta \cdot z}du}}{h_{0,\tau,\theta}(t)^\delta 
	(e^{\beta \cdot z})^\delta e^{-\int_\tau^t h_{0,\tau,\theta}(u)e^{\beta \cdot z}du}}\\
&= \delta \ln \frac{h_0(t)}{(\theta t)^{-1}} - e^{\beta \cdot z} \int_\tau^t(h_0(u) - \frac{1}{\theta u})du.
\end{align*}
It follows
\begin{align*}
\left| \ln\frac{dP_{S_0}(\given{t,\delta}{z})}{dP_{S_{0,\tau,\theta}}(\given{t,\delta}{z})} \right| &= \left|\delta \ln \frac{h_0(t)}{(\theta t)^{-1}} - \int_\tau^t\left(h(\given{u}{z}) - \frac{1}{\theta u}e^{\beta \cdot z}\right)du \right|\\
&= \left| \delta \ln( \frac{th_0(t)}{\theta^{-1}} -1+1) - \int_\tau^t\left(uh(\given{u}{z}) - \frac{1}{\theta}e^{\beta \cdot z}\right)\frac{du}{u}  \right|\\
& \leq   \delta 2\theta \left|th_0(t) - \frac{1}{\theta} \right| + \left| \int_\tau^t\left(uh(\given{u}{z}) - \frac{1}{\theta}e^{\beta \cdot z}\right)\frac{du}{u} \right|.
\end{align*}
Let $\rho_\tau = \sup_{t > \tau} \left| th_0(t) - \frac{1}{\theta} \right|$, then
$$
\left| \ln\frac{dP_{S_0}(\given{t,\delta}{z})}{dP_{S_{0,\tau,\theta}}(\given{t,\delta}{z})} \right| \leq \rho_\tau \left( \delta 2\theta + e^{\beta \cdot z} \ln\frac{t}{\tau} \right).
$$

We have 
\begin{align*}
&\max_{z \in \mathbb{Z}} \chi^2(P_{S_0}(\given{\cdot}{z}),P_{S_{0,\tau,\theta}}(\given{\cdot}{z})) \\
&\leq \max_{z \in \mathbb{Z}} \int \left( \ln\frac{dP_{S_0}(\given{\cdot}{z})}{dP_{S_{0,\tau,\theta}}(\given{\cdot}{z})} \right)^2 
\exp\left| \ln \frac{dP_{S_0}(\given{\cdot}{z})}{dP_{S_{0,\tau,\theta}}(\given{\cdot}{z})} \right| dP_{S_0}(\given{\cdot}{z})\\
&\leq \max_{z \in \mathbb{Z}} \int_{\tau}^\infty \sum_{\delta \in \{0,1\}}  \rho_\tau^2 \left( \delta 2\theta +  e^{\beta \cdot z}\ln\frac{t}{\tau} \right)^2 \exp( \rho_\tau \left( \delta 2\theta + e^{\beta \cdot z} \ln\frac{t}{\tau} \right)) P_{S_0}(\given{t,\delta}{z}) dt\\
&\leq \max_{z \in \mathbb{Z}} \int_{\tau}^\infty  \rho_\tau^2e^{2\beta \cdot z} \ln^2\frac{t}{\tau} \exp\left( \rho_\tau e^{\beta \cdot z}\ln\frac{t}{\tau} \right) S(\given{t}{z})f_C(\given{t}{z}) \\
& + \rho_\tau^2\left( 2\theta + e^{\beta \cdot z}\ln\frac{t}{\tau} \right)^2\exp\left(  \rho_\tau(2\theta+ e^{\beta \cdot z}\ln\frac{t}{\tau} ) \right)f(\given{t}{z})S_C(\given{t}{z}) dt\\
&\leq \max_{z \in \mathbb{Z}} \int_{\tau}^\infty \rho_\tau^2 \left( \frac{t}{\tau}  \right)^{\rho_\tau} 
\left( e^{2\beta \cdot z}\ln^2\frac{t}{\tau} + \left( 2\theta + e^{\beta \cdot z}\ln\frac{t}{\tau} \right)^2e^{2\rho_\tau\theta} \right)\\
&\qquad\qquad\qquad\qquad \times \left( f(\given{t}{z})S_C(\given{t}{z}) + S(\given{t}{z})f_C(\given{t}{z}) \right) dt.
\end{align*}
Let $g(u)=\left( e^{2\beta \cdot z}u^2 + (2\theta+e^{\beta \cdot z}u)^2e^{2\rho_\tau\theta}\right)e^{\rho_\tau u}$.
Then 
\begin{align*}
&\max_{z \in \mathbb{Z}}\chi^2(P_{S_0}(\given{.}{z}),P_{S_{0,\tau,\theta}}(\given{.}{z})) \\
&\leq \max_{z \in \mathbb{Z}}\int_{\tau}^\infty \rho_\tau^2 g\left(\ln\frac{t}{\tau}\right)(f(\given{t}{z}) S_C(\given{t}{z}) + S(\given{t}{z}) f_C(\given{t}{z})) dt.
\end{align*}
Since $S(\given{t}{z}) \leq S(\given{\tau}{z})$ and $S_C(\given{t}{z}) \leq S_C(\given{\tau}{z})$ for every $t > \tau$, we obtain
\begin{align*}
&\max_{z \in \mathbb{Z}}\chi^2(P_{S_0}(\given{.}{z}),P_{S_{0,\tau,\theta}}(\given{.}{z})) \\
&\leq \max_{z \in \mathbb{Z}} 
\rho_\tau^2  
\int_{\tau}^\infty g\left(\ln\frac{t}{\tau}\right)\\
& \qquad\qquad\times  \left( \frac{f(\given{t}{z}) }{S(\given{\tau}{z}) } S(\given{\tau}{z}) S_C(\given{\tau}{z}) +  \frac{f_C(\given{t}{z}) }{S_C(\given{\tau}{z}) } S(\given{\tau}{z}) S_C(\given{\tau}{z}) \right) dt\\
& \leq \max_{z \in \mathbb{Z}} \rho_\tau^2 S(\given{\tau}{z}) S_C(\given{\tau}{z}) \int_{\tau}^\infty g\left(\ln\frac{t}{\tau}\right)\left( \frac{f(\given{t}{z}) }{S(\given{\tau}{z}) }  +  \frac{f_C(\given{t}{z}) }{S_C(\given{\tau}{z}) } \right) dt.
\end{align*}
We can rewrite the ratio $\frac{S(\given{t}{z}) }{S(\given{\tau}{z}) }$ as $e^{-\int_{\tau}^t h(\given{u}{z})du} = e^{-\int_{\tau}^t uh(\given{u}{z})\frac{du}{u} }$. We know that $th(\given{t}{z})$ is bounded below for $t$ large enough by : $th(\given{t}{z}) \geq \frac{e^{\beta \cdot z}}{2 \theta}$. Then, 
\begin{equation*}
\frac{S(\given{t}{z}) }{S(\given{\tau}{z}) } \leq e^{-\frac{e^{\beta \cdot z}}{2\theta} \ln\frac{t}{\tau}}.
\end{equation*}
We know that :
\begin{align*}
\frac{d}{dt} g\left(\ln\frac{t}{\tau}\right)
&= \frac{d}{dt}  \left( \ln^2\frac{t}{\tau}+ \left(  2\theta + \ln\frac{t}{\tau} \right)^2 e^{2\rho_\tau\theta} \right) e^{\rho_\tau \ln\frac{t}{\tau}}\\
& = \left[ \frac{2}{t}\ln\frac{t}{\tau} + e^{2\rho_\tau\theta} \frac{2}{t} \left( 2\theta + \ln\frac{t}{\tau} \right) \right] e^{\rho_\tau \ln\frac{t}{\tau}}\\
&+ \left( \ln^2\frac{t}{\tau}+ \left(  2\theta + \ln\frac{t}{\tau} \right)^2 e^{2\rho_\tau\theta} \right)e^{\rho_\tau \ln\frac{t}{\tau}}\rho_\tau\frac{1}{t}\\
&= \left( \frac{t}{\tau} \right)^{\rho_\tau -1}\Bigg[ \frac{2}{\tau} \ln\frac{t}{\tau} + \frac{2}{\tau}e^{2\rho_\tau\theta}2\theta + \frac{2}{\tau}e^{2\rho_\tau\theta} \ln\frac{t}{\tau} + \frac{\rho_\tau}{\tau}\ln^2\frac{t}{\tau}\\
& \qquad+ \frac{\rho_\tau}{\tau}4\theta^2e^{2\rho_\tau\theta} +\frac{\rho_\tau}{\tau}4\theta\ln\frac{t}{\tau}e^{2\rho_\tau\theta} +\frac{\rho_\tau}{\tau}\ln^2\frac{t}{\tau}e^{2\rho_\tau\theta} \Bigg].
\end{align*}

Integrating by parts the term $\int_{\tau}^\infty g\left(\ln\frac{t}{\tau}\right)\left( \frac{f(\given{t}{z}) }{S(\given{\tau}{z}) } \right)dt$, we have :
\begin{align*}
&\int_{\tau}^\infty g \left(\ln\frac{t}{\tau}\right) \frac{f(\given{t}{z}) }{S(\given{\tau}{z}) } dt \\
& =  \left[ -g\left(\ln\frac{t}{\tau}\right) \frac{S(\given{t}{z}) }{S(\given{\tau}{z}) } \right]^\infty_\tau + \int_{\tau}^\infty \frac{S(\given{t}{z}) }{S(\given{\tau}{z}) } g'\left(\ln\frac{t}{\tau}\right) dx\\
& \leq 4\theta^2e^{2\rho_\tau\theta}+ \int_{\tau}^\infty e^{-\frac{e^{\beta \cdot z}}{2\theta} \ln\frac{t}{\tau}} g'( \ln\frac{t}{\tau}) dt\\
& \leq 4\theta^2e^{2\rho_\tau\theta}  + \int_{\tau}^\infty \left( \frac{t}{\tau} \right)^{\rho_\tau -\frac{e^{\beta \cdot z}}{2\theta} -1} \cdot \Bigg[ \frac{2}{\tau} \ln\frac{t}{\tau} + \frac{2}{\tau}e^{2\rho_\tau\theta}2\theta + \frac{2}{\tau}e^{2\rho_\tau\theta} \ln\frac{t}{\tau}\\ 
&\qquad\qquad\qquad + \frac{\rho_\tau}{\tau}\ln^2\frac{t}{\tau} + \frac{\rho_\tau}{\tau}4\theta^2e^{2\rho_\tau\theta} 
+\frac{\rho_\tau}{\tau}4\theta\ln\frac{t}{\tau}e^{2\rho_\tau\theta} +\frac{\rho_\tau}{\tau}\ln^2\frac{t}{\tau}e^{2\rho_\tau\theta} \Bigg] dt.
\end{align*}
We know that $\rho_\tau$ is supposed to be small for large values of $t$. It is safe to say that $t^{\rho_\tau -\frac{e^{\beta \cdot z}}{2\theta} -1} \leq t^{-\frac{1}{4\theta}}$. Then, $\int_{\tau}^\infty g\left(\ln\frac{t}{\tau}\right) \frac{f(\given{t}{z}) }{S(\given{\tau}{z}) } dt $ can be bounded by a constant.
$$
\int_{\tau}^\infty g\left(\ln\frac{t}{\tau}\right)  \frac{f(\given{t}{z}) }{S(\given{\tau}{z}) } dt = O_{\mathbb{P}}(1).
$$
In the same way, for $h_C(\given{t}{z})\geq c''$, for $t\geq \tau$, we have
$$
\int_{\tau}^\infty g\left(\ln\frac{t}{\tau}\right) \frac{f_C(\given{t}{z}) }{S_C(\given{t}{z}) } dt
= O_{\mathbb{P}}(1).$$
Then:
$$
\max_{z \in \mathbb{Z}} \chi^2(P_{S_0}(\given{\cdot}{z}),P_{S_{0,\tau,\theta}}(\given{\cdot}{z})) 
\leq O_{\mathbb{P}}\left( \rho_\tau^2 \max_{z \in \mathbb{Z}} S(\given{\tau}{z}) S_C(\given{\tau}{z}) \right).
$$
%%%%%%%%%%%%%%
%\todo{Correction $\tau$ to $\tau_n$}
%%%%%%%%%%%%%%%
Using the previous bound with $\tau=\tau_n$, we have, as $n \rightarrow \infty$,
\begin{align}\label{proofeq001}
\sum_{i=1}^{n}\chi^2(P_{S_0}(\given{\cdot}{z_i}),P_{S_{0,\tau_n,\theta}}(\given{\cdot}{z_i})) \leq n \rho_{\tau_n}^2 
\max_{z \in \mathbb{Z}} S(\given{\tau_n}{z}) S_C(\given{\tau_n}{z}) \rightarrow 0,
\end{align}
since $\max_{z \in \mathbb{Z}} S(\given{\tau_n}{z}) S_C(\given{\tau_n}{z})\leq 1$ and by condition \eqref{condMises00A}, 
$$
n \rho_{\tau_n}^2 \rightarrow 0 \quad \text{as} \quad n \rightarrow \infty.
$$
On the other hand, from condition \eqref{denominator001}, we have 
\begin{align*} \label{proofeq002}
\sum_{i=1}^nS_C(\given{\tau_n}{z_i})S(\given{\tau_n}{z_i}) \rightarrow \infty \quad \text{as} \quad n \rightarrow \infty.
\end{align*}
Since, by Theorem \ref{th1}, we have:
$$\lim_{n\to \infty}  \mathbb{P}  \left(\mathcal{K}(\widehat \theta_{\tau_n},\theta)  
\leq c \frac{  \sum_{i=1}^{n}\chi^2(P_{S_0}(\given{\cdot}{z_i}),P_{S_{0,\tau_n,\theta}}(\given{\cdot}{z_i}))+ 4\ln(n)}{\sum_{i=1}^nS_C(\given{\tau_n}{z_i})S(\given{\tau_n}{z_i})}     \right) 
= 1,$$
using \eqref{KL entropy}, it is easy to see that $\mathcal{K}(\widehat \theta_{\tau_n},\theta)  \rightarrow 0$ as $n \rightarrow \infty$, which means that 
$$
\widehat \theta_{\tau_n} \xrightarrow[n\rightarrow \infty]{\mathbb{P}} \theta .
$$

\subsection{Proof of Theorem \ref{Th5}}

\begin{proof}
	Starting from the auxiliary result in the proof for the Theorem \ref{th2}, we have
	$$
	\max_{z \in \mathbb{Z}} \chi^2(P_{S_0}(\given{\cdot}{z}),P_{S_{0,\tau,\theta}}(\given{\cdot}{z})) 
	\leq O_{\mathbb{P}}\left(\rho_\tau^2 \max_{z \in \mathbb{Z}}  S(\given{\tau}{z}) S_C(\given{\tau}{z}) \right)\quad \text{as } \tau \rightarrow \infty.
	$$
	We now want to find a sequence of threshold $(\tau_n)$ such that 
	$$
	\max_{z \in \mathbb{Z}}  S(\given{\tau_n}{z})S_C(\given{\tau_n}{z}) \rho_{\tau_n}^2 \leq c_0 \left( \frac{\ln n}{n} \right).
	$$
	%which implies that we assume the following condition
	Suppose that there exists a sequence of $(\tau_n)$ and a constant $c_0$ such that %, for any $z$ 
	\begin{equation}
	\label{c1}
	\max_{z \in \mathbb{Z}}  \chi^2(P_{S_0}(\given{.}{z}),P_{S_{0,\tau_n,\theta_{\tau_n}}}(\given{.}{z})) = c_0 \frac{\ln n}{n} .
	\end{equation}
	%Since $S_C(\given{\tau_n}{z}) \leq 1$, the previous expression becomes
	%\begin{equation}\label{proofTau}
	%S(\given{\tau_n}{z})S_C(\given{\tau_n}{z}) \rho_{\tau_n}^2 \leq c_0\left( \frac{\ln n}{n} \right).
	%\end{equation}
	Since the baseline hazard function is assumed to satisfy
	% model is related to the families of distributions for the extreme value model, we have the relation 
	(\ref{cHall1}),
	%$$
	%|xh_0(x) - \frac{1}{\theta}|\leq c_1x^{-\alpha}.
	%$$
	we have, 
	\begin{equation*}
	|xh(\given{x}{z}) - \frac{1}{\theta}e^{\beta \cdot z}|\leq c_1'x^{-\frac{\alpha e^{\beta \cdot z}}{\theta}}.
	\end{equation*}
	From (\ref{cHall1}), we find the following lower bound for $xh(\given{x}{z})$:
	%\begin{align*}
	% xh_0(x)- \frac{1}{\theta} &\geq -|xh_0(x) - \frac{1}{\theta}| 
	%\end{align*}
	%and 
	\begin{align*}
	xh(\given{x}{z}) &\geq \frac{e^{\beta \cdot z}}{\theta}-|xh(\given{x}{z}) - \frac{e^{\beta \cdot z}}{\theta}| \\
	&\geq \frac{e^{\beta \cdot z}}{\theta}-c_1x^{-\frac{\alpha e^{\beta \cdot z}}{\theta}} .
	\end{align*}
	We note that 
	\begin{align*}
	\max_{z \in \mathbb{Z}}  S(\given{\tau_n}{z}) &= \max_{z \in \mathbb{Z}}  \exp \left( - \int_{x_0}^{\tau_n} h(\given{t}{z})dt \right)\\
	& =\max_{z \in \mathbb{Z}}  \exp \left( - \int_{x_0}^{\tau_n} th(\given{t}{z})\frac{dt}{t} \right)\\
	& \leq\max_{z \in \mathbb{Z}}  \exp \left( - \frac{e^{\beta \cdot z}}{ \theta} \ln \frac{\tau_n}{x_0} - \frac{\theta c_1'}{\alpha e^{\beta \cdot z}}\left(\tau_n^{-\frac{\alpha e^{\beta \cdot z}}{\theta}} - x_0^{-\frac{\alpha e^{\beta \cdot z}}{\theta}}\right)\right)\\
	& \leq\max_{z \in \mathbb{Z}}  \exp \left( - \frac{e^{\beta \cdot z}}{\theta} \ln \frac{\tau_n}{x_0} \right) \exp \left(\frac{c_1'}{\alpha}x_0^{-\alpha}\right).
	\end{align*}
	We now find bounds for $xh_C(\given{x}{z})$.
	\begin{align*}
	xh_C(\given{x}{z})-  \frac{\gamma e^{\beta \cdot z}}{\theta} &\geq -|xh_C(\given{x}{z}) -  \frac{\gamma e^{\beta \cdot z}}{\theta}| \\
	xh_C(\given{x}{z}) &\geq  \frac{\gamma e^{\beta \cdot z}}{\theta}-|xh_C(\given{x}{z}) -  \frac{\gamma e^{\beta \cdot z}}{\theta}| \\
	&\geq  \frac{\gamma e^{\beta \cdot z}}{\theta}-c_2x^{-\mu} .
	\end{align*}
	We have 
	\begin{align*}
	\max_{z \in \mathbb{Z}}  S_C(\given{\tau_n}{z}) & =  \max_{z \in \mathbb{Z}} \exp \left( - \int_{x_0}^{\tau_n} h_C(\given{x}{z})dx \right)\\
	& = \max_{z \in \mathbb{Z}}  \exp \left( - \int_{x_0}^{\tau_n} xh_C(\given{x}{z}) \frac{dx}{x} \right)\\
	& \leq  \max_{z \in \mathbb{Z}} \exp \left( -  \frac{\gamma e^{\beta \cdot z}}{\theta} \ln \frac{\tau_n}{x_0} - \frac{c_2}{\mu}\left(\tau_n^{-\mu} - x_0^{-\mu}\right)\right)\\
	& \leq \max_{z \in \mathbb{Z}}  \exp \left( -  \frac{\gamma e^{\beta \cdot z}}{\theta} \ln \frac{\tau_n}{x_0} \right) \exp \left(\frac{c_2}{\mu}x_0^{-\mu}\right).
	\end{align*}
	Then,
	\begin{equation}\label{proofTau}
	\max_{z \in \mathbb{Z}}  S(\given{\tau_n}{z})S_C(\given{\tau_n}{z})\rho_{\tau_n}^2 \leq \max_{z \in \mathbb{Z}}  (c_1')^2 \tau_n^{-(\frac{e^{\beta \cdot z}}{ \theta}+ \frac{\gamma e^{\beta \cdot z}}{\theta}+2\frac{\alpha e^{\beta \cdot z}}{\theta} )}x_0^{\frac{e^{\beta \cdot z}}{ \theta}+ \frac{\gamma e^{\beta \cdot z}}{\theta}}c_3.
	\end{equation}
	Solving the equation (\ref{proofTau}) for $\tau_n$ yields
	\begin{align*}
	&\max_{z \in \mathbb{Z}}  (c_1')^2 \tau_n^{-(\frac{e^{\beta \cdot z}}{ \theta}+ \frac{\gamma e^{\beta \cdot z}}{\theta}+2\frac{\alpha e^{\beta \cdot z}}{\theta} )}x_0^{\frac{e^{\beta \cdot z}}{ \theta}+ \frac{\gamma e^{\beta \cdot z}}{\theta}}c_3= c_0\left( \frac{\ln n}{n} \right), \\
	&\tau_n =  n^{\frac{\theta}{\underset{z \in \mathbb{Z}}{\min} (e^{\beta \cdot z}) (1+\gamma +2\alpha)}} \ln^{-\frac{\theta}{\underset{z \in \mathbb{Z}}{\min} (e^{\beta \cdot z}) (1+\gamma +2\alpha)}} n, \\
	&\tau_n =  n^{\frac{\sfrac{\theta}{ \underset{z \in \mathbb{Z}}{\min} (e^{\beta \cdot z})}}{1+\gamma+2\alpha}} \ln^{-\frac{\sfrac{\theta}{ \underset{z \in \mathbb{Z}}{\min} (e^{\beta \cdot z})}}{1+\gamma +2\alpha}} n.
	\end{align*}
	Now, we search a lower bound for $\sum_{i=1}^nS_(\given{\tau}{z_i})S_C(\given{\tau_n}{z_i}) $. We have for any $z$,
	\begin{align*}
	S(\given{\tau_n}{z}) &= \exp\left( - \int_{x_0}^{\tau_n} th(\given{t}{z})\frac{dt}{t}\right) \\
	&\geq  \exp\left( - \frac{e^{\beta \cdot z}}{\theta} \ln\frac{\tau_n}{x_0} - \frac{\theta c_1}{-\alpha e^{\beta \cdot z}}\left(\tau_n^{-\frac{\alpha e^{\beta \cdot z}}{\theta}} - x_0^{-\frac{\alpha e^{\beta \cdot z}}{\theta}}  \right)\right)\\
	&\geq  \exp \left( - \frac{e^{\beta \cdot z}}{\theta} \ln \frac{\tau_n}{x_0} \right) \exp \left(-\frac{\theta c_1'}{\alpha e^{\beta \cdot z}}x_0^{-\frac{\alpha e^{\beta \cdot z}}{\theta}}\right),
	\end{align*}
	and
	\begin{align*}
	S_C(\given{\tau_n}{z}) &= \exp \left( - \int_{x_0}^{\tau_n} xh_C(x) \frac{dx}{x} \right) \\
	&\geq  \exp \left( -  \frac{\gamma e^{\beta \cdot z}}{\theta} \ln \frac{\tau_n}{x_0} \right) \exp \left(-\frac{c_2}{\mu}x_0^{-\mu}\right).
	\end{align*}
	Choosing $z$ such as minimizing $\sum_{i=1}^nS_(\given{\tau}{z_i})S_C(\given{\tau_n}{z_i}) $ yields the result of the theorem.
	%From the previous results, we have
	%\begin{align*}
	%\sum_{i=1}^nS(\given{\tau_n}{z_i})S_C(\given{\tau_n}{z_i}) &\geq c_4 \sum_{i=1}^n\exp\left( - \frac{e^{\beta \cdot z_i}}{\theta} \ln\frac{\tau_n}{x_0} \right)\exp\left( - \frac{\gamma e^{\beta \cdot z_i}}{\theta} \ln\frac{\tau_n}{x_0} \right) \\
	% &\geq  c_4\sum_{i=1}^n\exp\left( - \left(\frac{e^{\beta \cdot z_i}}{\theta}+\frac{\gamma e^{\beta \cdot z_i}}{\theta}\right) \ln\frac{ n^{\frac{\theta / e^{\beta \cdot z_i}}{1+\gamma+2\alpha}} \ln^{-\frac{\theta / e^{\beta \cdot z_i}}{1+\gamma+2\alpha}} n}{x_0} \right)\\
	%  &\geq  c_4\sum_{i=1}^n x_0^{\frac{e^{\beta \cdot z_i}}{\theta}+\frac{\gamma e^{\beta \cdot z_i}}{\theta}} \ln^{\frac{1+\gamma}{1+\gamma+2\alpha}}(n) n^{-\frac{1+\gamma}{1+\gamma+2\alpha}} \\
	% &\geq  c_0' n^{1-\frac{1+\gamma}{1+\gamma+2\alpha}}  \ln^{\frac{1+\gamma}{1+\gamma+2\alpha}}(n)\\
	%  &\geq  c_0' n^{\frac{2\alpha}{1+\gamma+2\alpha}}  \ln^{\frac{1+\gamma}{1+\gamma+2\alpha}}(n),
	%\end{align*}
	%where $c_0'$ is a positive constant.
\end{proof}

\bibliographystyle{tfnlm}
\bibliography{Cox_adaptive_varXiv}

\end{document}